\documentclass[preprint,12pt,authoryear]{elsarticle}

\usepackage[margin=1in]{geometry}
\usepackage{amsmath,amssymb,amsthm,mathtools}
\usepackage{algorithm}
\usepackage{algpseudocode}
\usepackage{booktabs}
\usepackage{graphicx}
\usepackage{enumitem}
\usepackage{placeins}
\usepackage{natbib}
\usepackage[colorlinks=true,linkcolor=blue,citecolor=blue,urlcolor=blue]{hyperref}

\newcommand{\figplaceholder}[1]{%
  \fbox{\begin{minipage}[c][0.28\textheight][c]{0.92\textwidth}
  \centering Figure file not attached: \texttt{\detokenize{#1}}
  \end{minipage}}%
}
\newcommand{\maybeincludegraphics}[2][]{%
  \IfFileExists{#2}{\includegraphics[#1]{#2}}{\figplaceholder{#2}}%
}

\newtheorem{theorem}{Theorem}
\newtheorem{proposition}{Proposition}

\theoremstyle{definition}

\newcommand{\argmin}{\mathop{\mathrm arg\,min}}

\newcommand{\bbE}{\mathbb{E}}
\newcommand{\1}{\mathbf{1}}
\newcommand{\Unif}{\mathrm{Unif}}

\journal{Journal of Computational and Graphical Statistics}

\begin{document}

\begin{frontmatter}

\title{Calibrated Predictive Distributions from Sample-Based Generators}

\author[nccu]{Wen-Ting Wang (\texttt{egpivo@gmail.com})}
\author[nccu]{ShengLi Tzeng (\texttt{sltzeng@nchu.edu.tw})}
\author[ti]{Yu-Ting Fan (\texttt{tinafan8181@gmail.com})}
\author[as,nthu]{Hsin-Cheng Huang\corref{cor1} (\texttt{hchuang@stat.sinica.edu.tw})}

\address[nccu]{Department of Applied Mathematics and Graduate Institute of Statistics, National Chung Hsing University, Taiwan}
\address[ti]{Texas Instruments, Taiwan}
\address[as]{Institute of Statistical Science, Academia Sinica, Taiwan}
\address[nthu]{Institute of Statistics and Data Science, National Tsing Hua University, Taiwan}

\cortext[cor1]{Corresponding author}

\date{}

\begin{abstract}
\baselineskip=18pt
Conditional generative models, including diffusion models and ensemble forecasters, often produce predictive samples without a tractable likelihood representation. Such sample-based predictive distributions can be systematically biased and poorly calibrated. We propose bias-corrected conformal probability integral transform (PIT) calibration, a split-sample post-processing framework that outputs a calibrated predictive distribution rather than a single fixed-level prediction interval. The method first estimates an affine location-scale correction on a held-out bias split, then calibrates randomized PIT values using a conformal calibrator. The resulting predictive law is represented as a weighted empirical distribution on the generator order statistics, enabling the direct computation of threshold-coherent exceedance probabilities, arbitrary quantiles, highest-density intervals, expected tail losses, and calibrated resamples. In contrast, standard conformal prediction primarily provides fixed-level prediction sets or threshold decisions and does not directly estimate predictive probabilities or high-density regions. We establish finite-sample calibration in probability under exchangeability and show how an optional split-conformal wrapper based on a PIT-centrality score gives nested prediction intervals with finite-sample marginal coverage at user-specified levels. Simulation studies with controlled misspecification and a WeatherBench-2 precipitation-forecasting application demonstrate substantial improvements in probabilistic calibration and downstream distributional summaries relative to uncalibrated sample-based forecasts and interval-only conformal baselines.
\bigskip
\end{abstract}


\end{frontmatter}

\baselineskip=18pt
\section{Introduction}
\label{sec:intro}

Modern conditional generative models are increasingly used for predictive inference. Given covariates $x$, a generator produces samples $Y^{(1)},\dots,Y^{(m)} \sim \widehat{P}(\cdot | x)$ that are intended to approximate the conditional distribution of $Y$ given $x$. In many applications, the model is accessed only through such samples rather than through an explicit likelihood or a closed-form density. This setting arises, for example, in diffusion-based inverse problems, in simulation-based science where forward simulators produce realizations, and in probabilistic forecasting pipelines that rely on ensemble-style outputs. In principle, such models deliver full predictive distributions rather than only point predictions. In practice, however, limited training data, model misspecification, and optimization artifacts can introduce substantial systematic bias and probability miscalibration into $\widehat{P}(\cdot|x)$, which in turn undermines downstream uncertainty quantification.

Conformal prediction provides distribution-free predictive inference with finite-sample marginal validity under exchangeability \citep{vovk2005algorithmic,shafer2008tutorial,lei2018dfpi}. Classical conformal regression typically begins with a point predictor and a scalar nonconformity score, such as an absolute residual, and returns a prediction interval at a specified miscoverage level $\alpha$. In sample-only generative settings, one common workaround is to compress the generated sample cloud to a point summary and then apply a standard conformal method. This approach discards distributional information, including skewness, heteroskedasticity, and multimodality, that may be crucial for downstream probability calculations.

The distinction between a valid prediction set and a calibrated predictive distribution is consequential for downstream analysis. Many scientific and operational users do not only ask whether a future response lies in a fixed $90\%$ interval. They ask for threshold probabilities such as $P(Y>c|x)$, tail summaries such as $\mathbb E\{(Y-c)_+|x\}$, Bayes actions under asymmetric losses, or simulated inputs for a subsequent decision model. These are functionals of the predictive law. A fixed-level conformal interval can support a valid set-valued statement, and a one-sided conformal construction can certify a threshold decision at a chosen error level, but neither of these outputs determines a predictive cumulative distribution function (CDF) or a calibrated estimate of an arbitrary exceedance probability.

Recent work has begun to use conditional random samples from a generator directly for conformal inference. Notably, \citet{pmlr-v206-wang23n} develop probabilistic conformal prediction using conditional random samples and construct prediction sets with marginal coverage guarantees. Such methods are well-suited when the target is a set-valued prediction region at a single level. Our objective differs in that we aim to produce a calibrated predictive distribution rather than merely a prediction set at a fixed level.
A predictive cumulative distribution function (CDF) supports calibrated quantiles, equal-tail summaries, and highest-density intervals at all levels, enables calibrated resampling for downstream Monte Carlo decision-making, and can be evaluated using calibration diagnostics such as the Cram\'er-von Mises distance and proper scoring rules such as the continuous ranked probability score (CRPS) \citep{gneiting2007scoring}.

Conformal predictive systems (CPSs) and conformal calibrators provide a general framework for constructing predictive distributions that are calibrated in probability \citep{pmlr-v60-vovk17a,pmlr-v128-vovk20a}. Our method shares this objective but is designed for a sample-only setting, in which the base predictive distribution is unavailable in closed form and is represented only by draws from an implicit generator. Consequently, for a finite generator sample of size $m$, the empirical predictive CDF is both discrete and random. We address these features by developing a randomized probability integral transform (PIT) tailored to sample-based predictors, together with an efficient representation that assigns nonuniform weights to the generator’s order statistics. This representation enables rapid evaluation of calibrated quantiles and direct resampling from the calibrated predictive distribution, without retraining the underlying generator.

Distributional conformal prediction (DCP) calibrates an estimated conditional CDF $\widehat F_x(y)$ using PIT ranks and studies approximate conditional validity under additional regularity and consistency conditions
\citep{chernozhukov2021dcp}. Our setting differs in that $\widehat F_x(y)$ is typically not directly evaluable and is accessible only through generated samples. Consequently, we focus on finite-sample marginal calibration in probability in the sense of CPS. Extensions toward more localized forms of validity, such as localized conformal methods or Mondrian variants, can in principle be incorporated, but they are not our primary focus \citep{guan2023localized,pmlr-v152-bostrom21a}. This is consistent with known impossibility results for exact distribution-free conditional coverage \citep{barber2021limits}.

A practical distinction from both CPS-style calibration and DCP-style PIT adjustment is that we introduce an explicit bias-correction operation, estimated on a held-out split, prior to conformal calibration. In many applications of generative models, systematic location and scale distortions are the dominant sources of error. We address this with an affine correction in two forms: a simple global adjustment based on residual location and scale, and a more flexible $x$-dependent correction in which the residual mean and log-scale ratio are modeled as smooth functions of the covariates using generalized additive model (GAM) regressions on the bias split.

The proposed procedure, bias-corrected conformal PIT calibration (CPIT),%
\footnote{An open-source implementation is available on GitHub (\url{https://github.com/egpivo/bc-cpit}) and PyPI (\url{https://pypi.org/project/bc-cpit}).}
has two distributional outputs. First, it produces a calibrated predictive CDF represented as a weighted empirical distribution on adjusted generator order statistics. This is the main object of the paper. It can be queried for event probabilities, arbitrary quantiles, HDIs, tail expectations, proper scores, and calibrated resamples. Second, it provides a smoothed weighted CDF when full-support summaries are desired. When exact coverage is needed at specified levels, the CPIT CDF can also be used inside a PIT-centrality split-conformal wrapper that returns nested fixed-level intervals. For comparison, we adapt three interval-only split-conformal baselines to the sample-based forecasting setting, using empirical quantiles, scaled residuals, and empirical CRPS scores, allowing us to distinguish fixed-level coverage calibration from calibration of the full predictive distribution.

Our contributions are as follows.
\begin{enumerate}[leftmargin=1.6em]
\item We develop CPIT, a split-sample post-processing framework for sample-based conditional generators. The method combines affine location-scale bias correction with conformal calibration of randomized PIT
values, and applies to black-box predictive samples without requiring a tractable likelihood. The proposed method outputs a calibrated predictive distribution, not only a fixed-level prediction interval.

\item We establish finite-sample calibration in probability for the randomized calibrated PIT values under exchangeability. We also provide an optional PIT-centrality split-conformal wrapper that gives nested fixed-level prediction intervals with finite-sample marginal coverage, and give stability and approximation results for the smoothed weighted CPIT distribution.

\item We demonstrate, through controlled simulations and a WeatherBench~2 precipitation-forecasting application, that CPIT improves distributional calibration and supports downstream probabilistic analyses, including
threshold-coherent exceedance probabilities, expected exceedance severity, tail-risk summaries, and highest-density intervals.
\end{enumerate}

The rest of the paper is organized as follows. Section~\ref{sec:method} introduces the sample-based forecasting setup and the CPIT procedure, including bias correction, PIT calibration, the weighted empirical predictive distribution, and its smoothed version. Section~\ref{sec:score_baselines} presents fixed-level conformal prediction methods for sample-based forecasts, including an exact PIT-centrality wrapper for CPIT and the interval-only conformal baselines used in the empirical comparisons. Section~\ref{sec:theory} gives the calibration, fixed-level coverage, and stability theory. Section~\ref{sec:sim} presents simulation studies that separate global bias, covariate-dependent bias, and distributional shape misspecification. Section~\ref{sec:weatherbench} applies
CPIT to WeatherBench~2 precipitation forecasts and evaluates distributional calibration, interval summaries, exceedance probabilities, and threshold-coherence. Section~\ref{sec:discussion} concludes with limitations and directions for future work.

\section{Proposed Method}
\label{sec:method}

\subsection{Statistical setup}

Let $(X_i,Y_i)$, $i=1,\ldots,n$, denote observations, where
$X_i\in\mathcal X$ contains the covariates used for prediction and
$Y_i\in\mathbb R$ is the scalar response. We assume access to a fitted
sample-only predictive model. For any covariate value $x$, this model returns
an $m$-vector of predictive draws,
\[
\widehat{\boldsymbol Y}(x)
=
\bigl(\widehat Y^{(1)}(x),\ldots,\widehat Y^{(m)}(x)\bigr)
\sim \widehat Q_x^{(m)} .
\]
Here $\widehat Q_x^{(m)}$ denotes the joint law of the $m$ generated values
induced by the fitted model at $x$. The generator may be biased or
miscalibrated, and $\widehat Q_x^{(m)}$ need not match the true conditional
law of $Y\mid X=x$. In particular, we do not require the generated values to
be conditionally independent.

For the observed case \(i\), we write
\[
\widehat Y_i^{(b)}=\widehat Y^{(b)}(X_i),\qquad b=1,\ldots,m.
\]
Thus the basic forecast-response object is
\begin{equation}
\mathcal O_i=\big(X_i,Y_i,\widehat Y_i^{(1)},\ldots,\widehat Y_i^{(m)}\big).
\label{eq:O}
\end{equation}
If the fitted generator is randomized, its Monte Carlo randomness is included in \(\mathcal O_i\). The conformal guarantees below are conditional on the fitted generator and require exchangeability of the corresponding forecast-response objects across the calibration and test cases.

We use four disjoint splits, $\mathcal I_{\mathrm tr}$, $\mathcal I_{\mathrm bias}$, $\mathcal I_{\mathrm cal}$, and $\mathcal I_{\mathrm test}$, for training, bias correction, PIT calibration, and testing. 
When exact fixed-level interval coverage is desired in addition to the CPIT predictive CDF, Section~\ref{sec:score_baselines} describes a final split-conformal wrapper that uses an additional interval-calibration split $\mathcal I_{\mathrm int}$, or a held-out portion of the available calibration data. This optional interval split is not used to estimate the CPIT calibration map. All finite-sample statements below are conditional on the fitted generator and bias-correction rule.

\subsection{Location-scale bias correction}
\label{sec:bias}

Before applying the affine correction, we allow an optional monotone transformation of the response. Let $T(\cdot)$ be a fixed, strictly increasing transformation from the support of $Y$ to an interval on the real line, with inverse $T^{-1}(\cdot)$. Examples include the identity transformation, a log transformation with an offset, or a Box-Cox transformation. The purpose of $T(\cdot)$ is to place the outcome on a scale where location and scale corrections are more appropriate, for example, by reducing skewness.

\subsubsection{Global affine correction}

Let
\[
\widehat\mu(x)=\frac{1}{m}\sum_{j=1}^m T\big(\widehat Y^{(j)}(x)\big),\qquad\widehat\sigma^2(x)=\frac{1}{m-1}\sum_{j=1}^m\big\{T\big(\widehat Y^{(j)}(x)\big)-\widehat\mu(x)\big\}^2
\]
denote the sample mean and variance of the transformed generator samples. The global correction applies a constant location shift and a positive scale factor on the transformed scale,
\begin{equation*}
  T\big(\widehat Y^{(j)}_{\mathrm{adj}}(x)\big)=a_0+b_0\,T\big(\widehat Y^{(j)}(x)\big),\qquad b_0>0.
\end{equation*}

On the bias split $\mathcal{I}_{\mathrm{bias}}$, we estimate $(a_0,b_0)$ by the heteroskedastic Gaussian quasi-likelihood on the
transformed scale,
\begin{equation}
  (\widehat a_0,\widehat b_0)
  \in \argmin_{a_0\in\mathbb{R},\,b_0>0}
  \sum_{i\in\mathcal{I}_{\mathrm{bias}}}
  \left[2\log b_0+
  \frac{\{T(Y_i)-a_0-b_0\,\widehat\mu(X_i)\}^2}
       {b_0^2\,\widehat\sigma^2(X_i)}
  \right].
  \label{eq:global-objective-transformed}
\end{equation}
For fixed $b_0$, the minimizer in $a_0$ is
\begin{equation*}
  \widehat a_0(b_0)=
  \frac{\sum_{i\in\mathcal{I}_{\mathrm{bias}}}\big\{T(Y_i)-b_0\,\widehat\mu(X_i)\big\}\big/\widehat\sigma^2(X_i)}
  {\sum_{i\in\mathcal{I}_{\mathrm{bias}}}1\big/\widehat\sigma^2(X_i)},
\end{equation*}
so the criterion can be minimized by profiling over the single positive parameter $b_0$. The resulting corrected samples are therefore
\[
\widehat Y^{(j)}_{\mathrm{adj}}(x)=T^{-1}\big(\widehat a_0+\widehat b_0\,T\big(\widehat Y^{(j)}(x)\big)\big),\qquad j=1,\dots,m.
\]

\subsubsection{Covariate-dependent correction}

When the generator bias varies with the covariates, we allow the center and spread of the transformed sample distribution to depend on \(x\). Let \(\mu_{\mathrm{adj}}(x)\) denote the adjusted transformed mean and let \(h(x)\) denote the log-variance multiplier. We define
\begin{equation}
T\bigl(\widehat Y_{\mathrm{adj}}^{(j)}(x)\bigr)=\mu_{\mathrm{adj}}(x)+
\exp\{h(x)/2\}
\big\{
T\big(\widehat Y^{(j)}(x)\big)-\widehat\mu(x)
\big\},
\qquad j=1,\ldots,m.
\label{eq:samples_corrected_local_affine_transformed}
\end{equation}
For each fixed \(x\), this remains an affine transformation of the generated samples, with \(\mu_{\mathrm{adj}}(x)\) and \(\exp\{h(x)/2\}\) controlling their transformed center and spread, respectively.

Let $\boldsymbol\phi(x)=(\phi_1(x),\ldots,\phi_p(x))'$ be a low-dimensional feature map containing covariates or summaries of the generated sample. We use separate additive predictors for the adjusted center and the log-variance multiplier:
\begin{equation}
\mu_{\mathrm{adj}}(x)=\alpha_0+\beta_0\,\widehat\mu(x)
+\sum_{\ell=1}^p f_{m,\ell}\{\phi_\ell(x)\},
\qquad
h(x)=\gamma_0
+\sum_{\ell=1}^p f_{h,\ell}\{\phi_\ell(x)\}.
\label{eq:local_affine_map_transformed}
\end{equation}
Here $\beta_0$ calibrates the relation between the observed response and the generator
center, whereas $\exp(\gamma_0/2)$ is the baseline spread multiplier. The smooth terms describe covariate-dependent departures from these baseline
relationships and are centered over the bias-correction split for identifiability.

We estimate the mean and scale functions jointly by minimizing the penalized normal
quasi-likelihood
\begin{align}
(\widehat\mu_{\mathrm{adj}},\widehat h)
=&~ \operatorname*{arg\,min}_{\mu_{\mathrm{adj}},h}
\bigg\{\sum_{i\in\mathcal I_{\mathrm{bias}}}
\left[
h(X_i)+
\frac{\{T(Y_i)-\mu_{\mathrm{adj}}(X_i)\}^2}
     {\widehat\sigma^2(X_i)\exp\{h(X_i)\}}
\right]\notag\\
&~ +
\sum_{\ell=1}^p \lambda_{m,\ell}J_{m,\ell}(f_{m,\ell})
+
\sum_{\ell=1}^p \lambda_{h,\ell}J_{h,\ell}(f_{h,\ell})\bigg\},
\label{eq:local_fit_transformed}
\end{align}
where $J_{m,\ell}$ and $J_{h,\ell}$ are spline
roughness penalties induced by
\eqref{eq:local_affine_map_transformed}.  This criterion corresponds to the working model
\[
T(Y_i)\mid X_i
\approx
N\!\left(\mu_{\mathrm{adj}}(X_i),
\widehat\sigma^2(X_i)\exp\{h(X_i)\}\right).
\]

In practice, \eqref{eq:local_fit_transformed} can be optimized by alternating two standard
GAM updates.

\begin{itemize}
\item Given a current estimate $h^{(t)}(\cdot)$, update
$\mu_{\mathrm{adj}}(\cdot)$ by fitting a weighted Gaussian GAM with weights
\[
w_i^{(t)}
=
\left[\widehat\sigma^2(X_i)\exp\{h^{(t)}(X_i)\}\right]^{-1}.
\]
Equivalently,
\begin{equation*}
\mu_{\mathrm{adj}}^{(t+1)}(\cdot)
=
\operatorname*{arg\,min}_{\mu_{\mathrm{adj}}(\cdot)}
\sum_{i\in\mathcal I_{\mathrm{bias}}}
w_i^{(t)}
\{T(Y_i)-\mu_{\mathrm{adj}}(X_i)\}^2
+
\sum_{\ell=1}^p \lambda_{m,\ell}J_{m,\ell}(f_{m,\ell}).
\label{eq:m_gam_fit_transformed}
\end{equation*}

\item Given the updated mean fit $\mu_{\mathrm{adj}}^{(t+1)}(\cdot)$, form
\[
q_i^{(t+1)}
=
\frac{\{T(Y_i)-\mu_{\mathrm{adj}}^{(t+1)}(X_i)\}^2}
     {\widehat\sigma^2(X_i)},
\qquad i\in\mathcal I_{\mathrm{bias}}.
\]
Update $h(\cdot)$ by solving
\begin{equation*}
h^{(t+1)}(\cdot)
=
\operatorname*{arg\,min}_{h(\cdot)}
\sum_{i\in\mathcal I_{\mathrm{bias}}}
\left[q_i^{(t+1)}\exp\{-h(X_i)\}+h(X_i)\right]
+\sum_{\ell=1}^p \lambda_{h,\ell}J_{h,\ell}(f_{h,\ell}).
\label{eq:h_gam_fit_transformed}
\end{equation*}
This is equivalent to fitting a Gamma GAM with log link and response $q_i^{(t+1)}$, for
which
\[
\mathbb E\big(q_i^{(t+1)}\big|X_i\big)=\exp\{h(X_i)\}.
\]
\end{itemize}

We initialize the iteration at the global affine fit,
\[
\mu_{\mathrm{adj}}^{(0)}(x)
=\widehat a_0+\widehat b_0\,\widehat\mu(x),
\qquad
h^{(0)}(x)=2\log\widehat b_0,
\]
and alternate the mean and scale updates until the criterion in
\eqref{eq:local_fit_transformed} stabilizes.  With the resulting estimates
$\widehat\mu_{\mathrm{adj}}$ and $\widehat h$, the covariate-dependent corrected samples are
\begin{equation}
\widehat Y_{\mathrm{adj}}^{(j)}(x)
=
T^{-1}\!\left(
\widehat\mu_{\mathrm{adj}}(x)
+\exp\{\widehat h(x)/2\}
\left\{
T\!\big(\widehat Y^{(j)}(x)\big)-\widehat\mu(x)
\right\}
\right),
\qquad j=1,\ldots,m.
\label{eq:local_b_transformed}
\end{equation}

With either the global or covariate-dependent transformed-scale correction, the bias-corrected empirical base CDF is
\begin{equation}
\widehat F_x^{\mathrm{adj}}(y)=\frac{1}{m}\sum_{j=1}^m\1\bigl\{\widehat Y^{(j)}_{\mathrm{adj}}(x)\le y\bigr\}.
\label{eq:cdf_corr}
\end{equation}
This distribution serves as the base predictive CDF to be calibrated in the next step.

\subsection{Conformal calibrator and calibrated predictive CDF}
\label{sec:calibrator}

For each calibration pair $(X_i,Y_i)$ with $i\in\mathcal I_{\mathrm{cal}}$, we draw $m$ samples from the generator and apply the bias correction from Section~\ref{sec:bias}. Let
$\hat Y^{(1)}_{i,\mathrm{adj}},\dots,\hat Y^{(m)}_{i,\mathrm{adj}}$
denote the resulting adjusted samples at $X_i$. Define
\begin{equation*}
N_i=\#\{j:\hat Y^{(j)}_{i,\mathrm{adj}}<Y_i\},
\qquad
N_i^*=\#\{j:\hat Y^{(j)}_{i,\mathrm{adj}}=Y_i\}.
\end{equation*}
The rank-cell randomized PIT is
\begin{equation}
u_i=\frac{N_i+V_i(N_i^*+1)}{m+1},
\qquad V_i\sim\Unif(0,1),
\label{eq:pit_random}
\end{equation}
where the $V_i$ are independent across $i$ and independent of the data. When there are no ties, $N_i^*=0$ and
$u_i=(N_i+V_i)/(m+1)$.
Thus $u_i$ is uniformly randomized within the rank cell indexed by $N_i$. Under an ideal continuous generator for which
$\big\{Y_i,\hat Y^{(1)}_{i,\mathrm{adj}},\ldots,\hat Y^{(m)}_{i,\mathrm{adj}}\big\}$
are exchangeable, $N_i$ is uniform on $\{0,\ldots,m\}$, and hence $u_i\sim\Unif(0,1)$ exactly.

Using the calibration PIT values $\{u_i:i\in\mathcal I_{\mathrm{cal}}\}$, define the conformal calibration map
\begin{equation}
\widehat C(t)=
\begin{cases}0, & t=0,\\[4pt]
\displaystyle
\frac{1+\sum_{i\in\mathcal I_{\mathrm{cal}}}\1\{u_i\le t\}}{|\mathcal I_{\mathrm{cal}}|+1}, & 0<t\le 1.
\end{cases}
\label{eq:calibrator}
\end{equation}
The randomized PIT in \eqref{eq:pit_random} is defined on an
$(m+1)$-cell rank grid, whereas the adjusted empirical distribution is
supported on $m$ generated values. To align these two representations,
define the finite-support calibration map
\begin{equation}
\widehat C_m(t)
=
\frac{\widehat C(mt/(m+1))}
     {\widehat C(m/(m+1))},
\qquad 0\le t\le1.
\label{eq:finite_support_calibrator}
\end{equation}
Because $\widehat C_m$ is nondecreasing and satisfies
$\widehat C_m(0)=0$ and $\widehat C_m(1)=1$, it is a proper
calibration map on $[0,1]$.

For a new covariate value $x$, the CPIT predictive CDF is
\begin{equation}
\widetilde F_x(y)
=
\widehat C_m\big(\widehat F_x^{\mathrm{adj}}(y)\big).
\label{eq:bc_cpit_cdf}
\end{equation}
The rescaling by $m/(m+1)$ maps the empirical-CDF grid
$j/m$ to the rank-cell grid $j/(m+1)$, while the denominator
normalizes the resulting finite-support distribution.
When $\widehat C$ is the identity map, $\widehat C_m$ is also the
identity, and \eqref{eq:bc_cpit_cdf} reduces exactly to the adjusted
empirical CDF $\widehat F_x^{\mathrm{adj}}$ in
\eqref{eq:cdf_corr}. The finite-sample result in
Theorem~\ref{thm:bccpit} concerns the randomized conformal transform
$\widehat C(u)$. Equation~\eqref{eq:bc_cpit_cdf} is its deterministic
finite-support projection and is the predictive CDF used for
downstream distributional summaries.

Algorithm~\ref{alg:bccpit} summarizes the complete procedure.

\begin{algorithm}[t]
\caption{CPIT (Bias-Corrected Conformal PIT Calibration)}
\label{alg:bccpit}
\begin{algorithmic}[1]
\Require Training split, bias split $\mathcal I_{\mathrm{bias}}$, calibration split $\mathcal I_{\mathrm{cal}}$, generator sample size $m$
\State Fit generator $G$ on the training split
\State Estimate either $(\widehat a_0,\widehat b_0)$ by \eqref{eq:global-objective-transformed} or $\widehat \mu_{\mathrm{adj}}(\cdot),\widehat h(\cdot)$ by \eqref{eq:local_b_transformed} on $\mathcal I_{\mathrm{bias}}$
\For{$i\in\mathcal I_{\mathrm{cal}}$}
  \State Draw $\hat y^{(1:m)}(X_i)\sim \widehat{Q}^{(m)}_{X_i}$ and apply the bias correction
  \State Compute $u_i$ using \eqref{eq:pit_random}
\EndFor
\State Construct $\widehat C$ by \eqref{eq:calibrator}
\State For a new $x$, draw $\hat Y^{(1:m)}(x)$, apply the bias correction, and output the calibrated predictive CDF $\widetilde F_x$ with rank-cell weights $w_1,\ldots,w_m$ by \eqref{eq:bc_cpit_cdf}-\eqref{eq:weights}
\end{algorithmic}
\end{algorithm}

\subsection{Quantiles, highest-density intervals, and calibrated resampling}
\label{sec:quantiles_resampling}

A central advantage of CPIT is that it outputs a calibrated predictive distribution rather than only a prediction interval at a single nominal level. This section describes how to compute quantiles, equal-tail summaries, HDIs, calibrated resamples, and a smoothed calibrated CDF from \(\widetilde F_x\).

Let $\widehat{\boldsymbol Y}_{\!\!\mathrm{adj}}(x)=\{\widehat Y_{\mathrm{adj}}^{(1)}(x),\ldots,\widehat Y_{\mathrm{adj}}^{(m)}(x)\}$ be the bias-corrected generator samples at covariate value \(x\). We write
\begin{equation*}
\widehat Y_{\mathrm{adj},(1)}(x)\le\cdots\le\widehat Y_{\mathrm{adj},(m)}(x)
\label{eq:cdf_smooth}
\end{equation*}
for their order statistics, with ties counted according to their multiplicity. The normalized rank-cell weights are
\begin{equation}
w_j
=
\widehat C_m\left(\frac{j}{m}\right)
-
\widehat C_m\left(\frac{j-1}{m}\right),
\qquad j=1,\ldots,m.
\label{eq:weights}
\end{equation}
Because \(\widehat C\) is nondecreasing, the weights are nonnegative and sum to one. Thus the calibrated predictive distribution can be written as the weighted empirical measure
\begin{equation}
\widetilde P_x=\sum_{j=1}^m w_j\,\delta_{\widehat Y_{\mathrm{adj},(j)}(x)},
\label{eq:weighted measure}
\end{equation}
with CDF 
\begin{equation}
\widetilde F_x(y)=\sum_{j=1}^m w_j \,\1\big\{\widehat Y_{\mathrm{adj},(j)}(x)\le y\big\},
\label{eq:weighted-empirical}
\end{equation}
which is the same as that in \eqref{eq:bc_cpit_cdf}. The same telescoping argument remains valid when adjusted draws are tied, because all copies of a tied value are consecutive in the ordered sample. If the conformal calibration map is the identity, then  \(w_j=1/m\) for all \(j\). CPIT therefore reduces to the ordinary empirical distribution of the adjusted generator draws.

The corresponding quantile is
\[
\widetilde Q_x(\alpha)=\inf\{y:\widetilde F_x(y)\ge \alpha\},\qquad \alpha\in(0,1).
\]
An equal-tail \((1-\alpha)\) empirical CPIT interval is therefore
\begin{equation}
\widetilde I_\alpha(x)=\big[\widetilde Q_x(\alpha/2),\,\widetilde Q_x(1-\alpha/2)\big].
\label{eq:CI}
\end{equation}
The intervals in \eqref{eq:CI} are summaries of the calibrated predictive CDF rather than separate split-conformal prediction sets. Section~\ref{sec:score_baselines} shows how to add a final PIT-centrality conformal layer when exact fixed-level coverage is desired.

The equal-tail CPIT intervals in \eqref{eq:CI} are only one type of interval
that can be extracted from the calibrated predictive distribution. Because
CPIT returns a full weighted predictive law, we can also construct a
highest-density-region (HDR) interval, defined here as the shortest interval
whose calibrated predictive mass is at least \(1-\alpha\). With $w_0=0$, define
\begin{equation}
(k_\alpha,\ell_\alpha)\in\operatorname*{arg\,min}_{1\le k\le \ell\le m}
\Big\{
\widehat Y_{\mathrm{adj},(\ell)}(x)-\widehat Y_{\mathrm{adj},(k)}(x):w_k+\cdots+w_\ell\ge 1-\alpha
\big\},
\label{eq:hdi_indices}
\end{equation}
and the empirical CPIT HDR interval is
\begin{equation}
\widetilde H_\alpha(x)
=
\bigl[
\widehat Y_{\mathrm{adj},(k_\alpha)}(x),
\widehat Y_{\mathrm{adj},(\ell_\alpha)}(x)
\bigr].
\label{eq:hdi}
\end{equation}
After sorting, \eqref{eq:hdi_indices} can be computed by a one-pass
two-pointer search over the cumulative weights. This construction highlights a
key distinction between CPIT and interval-only conformal methods. Once a
calibrated predictive law is available, intervals can be chosen by optimizing
over calibrated probability mass, rather than being fixed in advance by a
single nonconformity score.

The same weights yield calibrated resampling.  To draw from $\widetilde P_x$, sample $J\in\{1,\ldots,m\}$ with probabilities $w_1,\ldots,w_m$ and return $\widehat Y_{\mathrm{adj},(j)}(x)$.  More generally,
\[
  \bbE_{\widetilde P_x}\{f(Y)\}=\sum_{j=1}^m w_j\,f\big(\widehat Y_{\mathrm{adj},(j)}(x)\big)
\]
for any measurable function $f$. This weighted representation is useful in downstream Monte Carlo tasks such as estimating exceedance probabilities, computing risk measures, or propagating predictive uncertainty through a decision rule.

\subsection{Smoothed weighted CPIT distribution}
\label{sec:smoothed-cpit}

The weighted empirical distribution in
\eqref{eq:weighted-empirical} is supported only on the adjusted
generator draws. To obtain a continuous predictive law with full
support, we smooth this distribution by Gaussian convolution. For a
bandwidth $\tau_x>0$, define
\begin{equation}
\widetilde F_x^\tau(y)
=
\sum_{j=1}^m
w_j
\Phi\bigg(
\frac{y-\widehat Y_{\mathrm{adj},(j)}(x)}
     {\tau_x}
\bigg),
\label{eq:smoothed-cdf}
\end{equation}
where $\Phi$ denotes the standard normal CDF. Let
$\widetilde P_x^\tau$ denote the probability distribution associated
with \eqref{eq:smoothed-cdf}. Because the weights are nonnegative and
sum to one, $\widetilde F_x^\tau$ is a proper continuous CDF. Moreover,
as $\tau_x\rightarrow0$,
$\widetilde F_x^\tau(y)$ converges to the weighted empirical CDF at
each of its continuity points.

For the Gaussian kernel in \eqref{eq:smoothed-cdf}, the
normal-reference rule for kernel distribution-function estimation
gives
\[
\tau_x
=
4^{1/3}\widehat s_x m^{-1/3},
\]
where $\widehat s_x$ is the unweighted sample standard deviation of the
adjusted generator draws
$\widehat Y_{\mathrm{adj}}^{(1)}(x),\ldots,
\widehat Y_{\mathrm{adj}}^{(m)}(x)$
\citep{lopez2015bandwidth}. Because this rule is derived for an
ordinary unweighted kernel CDF estimator, we use it as a simple
scale-adaptive default rather than as an optimal bandwidth for the
calibration-dependent CPIT weights. More data-adaptive alternatives
include plug-in bandwidth selection
\citep{altman1995bandwidth} and CDF-specific cross-validation
\citep{bowman1998bandwidth}.

The corresponding smoothed quantile function is
\[
\widetilde Q_x^\tau(q)
=
\inf\left\{
y:\widetilde F_x^\tau(y)\ge q
\right\},
\qquad 0<q<1.
\]
Hence, the smoothed equal-tail $(1-\alpha)$ CPIT interval is
\begin{equation}
\widetilde I_\alpha^\tau(x)
=
\left[
\widetilde Q_x^\tau(\alpha/2),
\widetilde Q_x^\tau(1-\alpha/2)
\right].
\label{eq:smoothed_CI}
\end{equation}

The density associated with \eqref{eq:smoothed-cdf} is the Gaussian
mixture
\[
\widetilde f_x^\tau(y)
=
\sum_{j=1}^m
w_j\tau_x^{-1}
\phi\bigg(
\frac{y-\widehat Y_{\mathrm{adj},(j)}(x)}
     {\tau_x}
\bigg),
\]
where $\phi$ is the standard normal density. A smoothed
highest-density region is
\[
\widetilde H_\alpha^\tau(x)
=
\left\{
y:
\widetilde f_x^\tau(y)\ge\lambda_\alpha(x)
\right\},
\]
where
\[
\lambda_\alpha(x)
=
\sup\left\{
\lambda\ge0:
\widetilde P_x^\tau
\big(
\widetilde f_x^\tau(Y)\ge\lambda
\big)
\ge1-\alpha
\right\}.
\]
For a multimodal predictive distribution,
$\widetilde H_\alpha^\tau(x)$ may be a union of disjoint intervals.
The smoothed distribution can be sampled by first drawing $J$ with probabilities $w_1,\ldots,w_m$ and then drawing from $N\big(\widehat Y_{\mathrm{adj},(J)}(x),\tau_x^2\big)$.

\section{Fixed-Level Conformal Prediction}
\label{sec:score_baselines}

This section collects fixed-level split-conformal constructions for sample-based predictive distributions. The first construction is an optional wrapper for CPIT. It uses the calibrated CDF only through a scalar PIT-centrality score and therefore gives the usual finite-sample marginal coverage guarantee at any chosen significance level. The remaining constructions are interval-only conformal baselines used in the empirical comparisons. They are useful fixed-level competitors, but their output is a prediction set at a selected nominal level rather than a calibrated predictive CDF.

\subsection{Generic split-conformal construction}
\label{sec:fixed_level_scores}

Let \(\mathcal I_{\mathrm conf}\) be a conformal calibration split that is not
used to fit the prediction rule or score function. For a fixed nonconformity
score \(s(x,y)\), compute
\[
S_i=s(X_i,Y_i), \qquad i\in\mathcal I_{\mathrm conf},
\]
and let \(S_{(1)}\le \cdots \le S_{(N_{\mathrm conf})}\) denote their order
statistics, where \(N_{\mathrm conf}=|\mathcal I_{\mathrm conf}|\). For a target
miscoverage level \(\alpha\), set
\[
k_\alpha=\left\lceil (N_{\mathrm conf}+1)(1-\alpha)\right\rceil,
\qquad
q_{\alpha,s}=\begin{cases}
S_{(k_\alpha)}, & k_\alpha\le N_{\mathrm conf},\\
+\infty, & k_\alpha>N_{\mathrm conf}.
\end{cases}
\]
The split-conformal prediction set is
\[
\mathcal C_{\alpha,s}(x)=\{y:s(x,y)\le q_{\alpha,s}\}.
\]
Under exchangeability of the conformal calibration cases and the test case,
this set has marginal coverage at least \(1-\alpha\), conditional on all data
used to construct the score. This generic construction will be used in two
ways. For the CPIT-centrality wrapper below, \(\mathcal I_{\mathrm conf}=\mathcal
I_{\mathrm int}\), which is separate from the PIT-calibration split. For the
interval-only baselines, \(\mathcal I_{\mathrm conf}=\mathcal I_{\mathrm cal}\) is the usual conformal
calibration split for the corresponding fixed-level score.

\subsection{PIT-centrality conformal intervals from CPIT}
\label{sec:pit_central_interval_method}

The equal-tail CPIT interval in \eqref{eq:CI} is a distributional summary of
\(\widetilde F_x\). If exact finite-sample coverage is required at a chosen
level, we can add a final split-conformal layer after the CPIT CDF has been
constructed. Let \(\mathcal I_{\mathrm int}\) be an interval-calibration split not
used to estimate the affine correction or the CPIT calibration map, and define
the PIT-centrality score
\begin{equation}
S_{\mathrm pit}(x,y)=\big|2\widetilde F_x(y)-1\big|.
\label{eq:pit-centrality-score}
\end{equation}
Let \(q_\alpha^{\mathrm pit}\) be the split-conformal quantile of
\(S_{\mathrm pit}(X_i,Y_i)\), \(i\in\mathcal I_{\mathrm int}\), computed as in
Section~\ref{sec:fixed_level_scores}. The resulting conformalized CPIT set is
\begin{equation}
\mathcal C_\alpha^{\mathrm pit}(x)
=
\left\{y:S_{\mathrm pit}(x,y)\le q_\alpha^{\mathrm pit}\right\}.
\label{eq:pit-centrality-set}
\end{equation}
For a continuous and strictly increasing $\widetilde F_x$, this set is the
interval
\begin{equation}
\mathcal C_\alpha^{\mathrm pit}(x)
=
\left[
\widetilde Q_x\big((1-q_\alpha^{\mathrm pit})/2\big),
\widetilde Q_x\big((1+q_\alpha^{\mathrm pit})/2\big)
\right].
\label{eq:pit-centrality-inverse}
\end{equation}
For a general right-continuous CDF, \eqref{eq:pit-centrality-set} is the exact
conformal set and should be used directly. A closed interval containing this
set can instead be formed from the lower endpoint
$\inf\{y:\widetilde F_x(y)\ge (1-q_\alpha^{\mathrm pit})/2\}$ and the upper endpoint
$\sup\{y:\widetilde F_x(y)\le (1+q_\alpha^{\mathrm pit})/2\}$. The same calibration
scores can be queried at several values of $\alpha$, producing a nested family
because $q_\alpha^{\mathrm pit}$ is monotone in $1-\alpha$.
Theorem~\ref{thm:pit_central_interval} gives the finite-sample coverage guarantee
for the exact score sublevel set at each chosen level. Thus CPIT supplies the
calibrated predictive CDF, while the PIT-centrality wrapper supplies fixed-level
marginal coverage when it is needed.

\subsection{Interval-only conformal baselines}
\label{sec:interval_baselines}

For the interval-only baselines, the split
$\mathcal I_{\mathrm{cal}}=\mathcal I_{\mathrm{conf}}$ is used directly as the split-conformal
calibration set. 
For each $i\in\mathcal I_{\mathrm{cal}}$, obtain the adjusted predictive
samples
\(
\widehat Y_{\mathrm{adj}}^{(1)}(X_i),\ldots,
\widehat Y_{\mathrm{adj}}^{(m)}(X_i).
\)
The same notation covers the Raw, Global, and GAM baselines by taking the
bias correction to be, respectively, the identity, the global correction,
or the covariate-dependent correction.

Recall the adjusted empirical CDF $\widehat F_x^{\mathrm{adj}}$ from \eqref{eq:cdf_corr}, and define its quantile function by
\[
\widehat Q_x^{\mathrm{adj}}(u)
=
\inf\left\{
y:
\widehat F_x^{\mathrm{adj}}(y)\ge u
\right\},
\qquad 0<u<1.
\]

\subsection*{Quantile residual score}

For a target miscoverage level \(\alpha\), the quantile residual score is
\begin{equation}
s_{\mathrm{QR},\alpha}(x,y)=\max\big\{\widehat U_{\mathrm{adj},\alpha/2}(x)-y,\,y-\widehat Q_{\mathrm{adj},(1-\alpha)/2}(x)\big\}.
\label{eq:qr_signed}
\end{equation}
This is the conformalized quantile regression score
\citep{romano2019cqr}, applied here to empirical quantiles of the
bias-adjusted generator samples. Because the score is signed, the conformal
correction may be negative. Thus the calibrated interval can shrink an overly
conservative base interval as well as expand an undercovering one.

Let \(q_\alpha^{\mathrm{QR}}\) be the split-conformal quantile computed from
the calibration scores \(s_{\mathrm{QR},\alpha}(X_i,Y_i)\). The resulting
prediction interval is
\[
\mathcal C_\alpha^{\mathrm{QR}}(x)
=
\big[
\widehat Q_{\mathrm{adj},\alpha/2}(x)-q_\alpha^{\mathrm{QR}},\,
\widehat Q_{\mathrm{adj},(1-\alpha)/2}(x)+q_\alpha^{\mathrm{QR}}
\big],
\]
with the convention that the interval is empty if the lower endpoint exceeds
the upper endpoint.

\subsection*{Scaled residual score}

The scaled residual score uses the empirical center and spread of the
bias-adjusted sample cloud:
\begin{equation}
s_{\mathrm{SR}}(x,y)
=
\frac{|y-\widehat\mu_{\mathrm{adj}}(x)|}
{\widehat\sigma_{\mathrm{adj}}(x)}.
\label{eq:scaled_resid}
\end{equation}
This score adapts the residual magnitude to the local dispersion of the
generator samples. Let \(q_\alpha^{\mathrm{SR}}\) be the split-conformal
quantile of the scaled residual scores. Inverting
\(s_{\mathrm{SR}}(x,y)\le q_\alpha^{\mathrm{SR}}\) gives
\[
\mathcal C_\alpha^{\mathrm{SR}}(x)
=
\big[
\widehat\mu_{\mathrm{adj}}(x)
-
q_\alpha^{\mathrm{SR}}\widehat\sigma_{\mathrm{adj}}(x),\,
\widehat\mu_{\mathrm{adj}}(x)
+
q_\alpha^{\mathrm{SR}}\widehat\sigma_{\mathrm{adj}}(x)
\big].
\]

\subsection*{Empirical CRPS score}

The empirical CRPS score evaluates the full bias-adjusted sample cloud rather
than only its center and spread \citep{gneiting2007scoring}:
\begin{equation}
s_{\mathrm{CRPS}}(x,y)
=
\frac{1}{m}\sum_{b=1}^m
|\widehat Y_{\mathrm{adj}}^{(b)}(x)-y|
-
\frac{1}{2m^2}
\sum_{b=1}^m
\sum_{b'=1}^m
|\widehat Y_{\mathrm{adj}}^{(b)}(x)
-\widehat Y_{\mathrm{adj}}^{(b')}(x)|.
\label{eq:crps_emp}
\end{equation}
For fixed samples
\(\widehat Y_{\mathrm{adj}}^{(1:m)}(x)\), the function \(s_{\mathrm{CRPS}}(x,y)\) is convex and piecewise linear in $y$. Therefore, the split-conformal prediction set
\[
\mathcal C_\alpha^{\mathrm{CRPS}}(x)
=
\{y:\ s_{\mathrm{CRPS}}(x,y)\le q_\alpha^{\mathrm{CRPS}}\}
\]
is an interval,
where \(q_\alpha^{\mathrm{CRPS}}\) is the split-conformal quantile of the empirical CRPS scores.


\section{Theory}
\label{sec:theory}

This section establishes the theoretical guarantees for CPIT and the optional
fixed-level construction in Section~\ref{sec:score_baselines}. We first prove a
finite-sample rank-calibration result for the conformally transformed randomized
PIT. We then establish marginal coverage and nesting for the PIT-centrality
split-conformal wrapper. Finally, we study the weighted CPIT law, showing that
Gaussian smoothing yields a full-support distribution and that estimation error
in the one-dimensional calibration map propagates stably to the predictive CDF.

Let $\mathcal D_{\mathrm tr}$ and $\mathcal D_{\mathrm bias}$ collect the data and
auxiliary randomness used to fit the generator and the bias-correction rule.
These quantities are treated as fixed in the conditional statements below. We
write $N=|\mathcal I_{\mathrm cal}|$.

\subsection{Finite-sample PIT calibration}

For a calibrated continuous predictive distribution, the randomized PIT is
uniform on $[0,1]$. Because CPIT starts from an empirical distribution supported
on generated values, its basic finite-sample statement is instead a conformal
rank result. The next theorem concerns the transformed value
$\widehat C(u_{n+1})$, not the raw randomized PIT $u_{n+1}$.

\begin{theorem}[Finite-sample CPIT rank calibration]
\label{thm:bccpit}
For every $i\in\mathcal I_{\mathrm cal}\cup\{n+1\}$, let $u_i$ be defined by
\eqref{eq:pit_random} using the auxiliary variable $V_i$. Suppose that,
conditional on $\mathcal D_{\mathrm tr}$ and $\mathcal D_{\mathrm bias}$,
\[
\{(\mathcal O_i,V_i):i\in\mathcal I_{\mathrm cal}\cup\{n+1\}\}
\]
is exchangeable, where $\mathcal O_i$ is defined in \eqref{eq:O}. Then the
calibration map $\widehat C$ in \eqref{eq:calibrator} satisfies
\begin{equation}
\mathbb P\!\left\{
\widehat C(u_{n+1})\le t
\,\middle|\,
\mathcal D_{\mathrm tr},\mathcal D_{\mathrm bias}
\right\}
\le t,
\qquad t\in[0,1].
\label{eq:finite_pit_calibration}
\end{equation}
Moreover, conditional on $\mathcal D_{\mathrm tr}$ and $\mathcal D_{\mathrm bias}$,
$\widehat C(u_{n+1})$ is discrete uniform on
$\{1/(N+1),\ldots,1\}$.
\end{theorem}

Theorem~\ref{thm:bccpit} is the finite-sample calibration-in-probability
statement for CPIT. The conformal transform of the future PIT is super-uniform,
and in fact exactly uniform on the conformal rank grid. The theorem does not
claim that the untransformed PIT $u_{n+1}$ is uniform under misspecification.

\subsection{Finite-sample fixed-level coverage from PIT centrality}
\label{sec:pit_central_interval_theory}

When coverage at a prespecified level is required, the CPIT CDF can be followed
by a separate split-conformal layer. This layer uses the CDF only through the
scalar PIT-centrality score, so its validity follows from the usual exchangeable
rank argument. The CPIT CDF remains the object used for probability, risk, and
other distributional summaries.

\begin{theorem}[PIT-centrality conformal sets]
\label{thm:pit_central_interval}
Let $\widetilde F_x$ be an empirical or smoothed CPIT predictive CDF constructed
without using the interval-calibration split $\mathcal I_{\mathrm int}$. Let
$\mathcal D_{\mathrm CPIT}$ denote all data and auxiliary randomness used to
construct the CPIT rule, excluding $\mathcal I_{\mathrm int}$ and the test case.
Suppose that, conditional on $\mathcal D_{\mathrm CPIT}$,
\[
\{\mathcal O_i:i\in\mathcal I_{\mathrm int}\cup\{n+1\}\}
\]
is exchangeable. For $i\in\mathcal I_{\mathrm int}$, define the exact score sublevel set by
\eqref{eq:pit-centrality-set}. Then
\begin{equation}
\mathbb P\!\left\{
Y_{n+1}\in\mathcal C_\alpha^{\mathrm pit}(X_{n+1})
\,\middle|\,
\mathcal D_{\mathrm CPIT}
\right\}
\ge 1-\alpha.
\label{eq:pit_central_coverage}
\end{equation}
If, in addition, the $N_{\mathrm int}+1$ calibration and test scores are almost
surely distinct conditional on $\mathcal D_{\mathrm CPIT}$, then
\[
\mathbb P\!\left\{
Y_{n+1}\in\mathcal C_\alpha^{\mathrm pit}(X_{n+1})
\,\middle|\,
\mathcal D_{\mathrm CPIT}
\right\}
\le 1-\alpha+\frac{1}{|\mathcal I_{\mathrm int}|+1}.
\]
Moreover, if $\alpha_1<\alpha_2$, then
$\mathcal C_{\alpha_1}^{\mathrm pit}(x)\supseteq
\mathcal C_{\alpha_2}^{\mathrm pit}(x)$ for every $x$.
\end{theorem}

For a continuous and strictly increasing $\widetilde F_x$, the exact score set
is the interval in \eqref{eq:pit-centrality-inverse}. For a discontinuous CDF,
the score sublevel set in \eqref{eq:pit-centrality-set} is the object covered by
the theorem; replacing it by a larger closed interval preserves the lower
coverage bound but can invalidate the upper bound. Reusing the same calibration
scores across several values of $\alpha$ gives nested sets with levelwise
marginal guarantees, not a simultaneous coverage statement for the entire
random family.

\subsection{Smoothed weighted CPIT CDF}

The weighted CPIT law $\widetilde P_x$ is discrete because it is supported on
the $m$ adjusted generator order statistics. Gaussian smoothing converts this
law into a continuous full-support distribution. We quantify the perturbation
using the 1-Wasserstein distance. For probability measures $P$ and $Q$ on
$\mathbb R$ with finite first moments,
\[
W_1(P,Q)
=
\inf_{\gamma\in\Gamma(P,Q)}
\int |u-v|\,d\gamma(u,v),
\]
where $\Gamma(P,Q)$ is the set of couplings of $P$ and $Q$. By the
Kantorovich--Rubinstein duality,
\[
W_1(P,Q)
=
\sup_{\|h\|_{\mathrm{Lip}}\le 1}
\left|
\mathbb E_P\{h(Y)\}-\mathbb E_Q\{h(Y)\}
\right|.
\]
Thus $W_1$ directly controls the error of Lipschitz downstream summaries (see, e.g., \citet{villani2009optimal}).

\begin{proposition}
\label{prop:smoothed_regular}
For any $x$ and $\tau_x>0$, the smoothed CPIT CDF
$\widetilde F_x^\tau$ in \eqref{eq:smoothed-cdf} is the CDF of an absolutely
continuous distribution with density
\begin{equation}
\widetilde f_x^\tau(y)
=
\sum_{j=1}^m
w_j\,\tau_x^{-1}
\phi\!\left(
\frac{y-\widehat Y_{\mathrm{adj},(j)}(x)}{\tau_x}
\right).
\label{eq:smoothed_density_theory}
\end{equation}
The distribution has support $\mathbb R$, and $\widetilde F_x^\tau$ is strictly
increasing. Let $\widetilde P_x$ be the weighted empirical distribution in
\eqref{eq:weighted measure}, and let $\widetilde P_x^\tau$ be the distribution
with CDF $\widetilde F_x^\tau$. Then
\begin{equation}
W_1(\widetilde P_x^\tau,\widetilde P_x)
\le
\tau_x\sqrt{2/\pi}.
\label{eq:wasserstein_smoothing}
\end{equation}
Consequently, for every $L$-Lipschitz function $h$,
\begin{equation}
\left|
\mathbb E_{\widetilde P_x^\tau}\{h(Y)\}
-
\mathbb E_{\widetilde P_x}\{h(Y)\}
\right|
\le
L\tau_x\sqrt{2/\pi}.
\label{eq:lipschitz_smoothing}
\end{equation}
\end{proposition}

Proposition~\ref{prop:smoothed_regular} is a stability statement rather than an
additional conformal guarantee. It shows that smoothing regularizes the
predictive law while perturbing any Lipschitz summary by at most a quantity
proportional to the bandwidth.

\subsection{Stability with respect to the calibration map}

The CPIT predictive distribution depends on the fitted calibration map only
through the normalized rank-cell weights in \eqref{eq:weights}. The next theorem
shows that uniform estimation of this one-dimensional map yields uniform control
of the resulting predictive CDF. The normalization by
$\widehat C\{m/(m+1)\}$ conditions the first $m$ rank-cell masses on the portion
of the grid represented by the $m$ generated order statistics. The remaining
rank cell, corresponding to a response above all generated values, has no
separate atom in the finite-support predictive law.

\begin{theorem}[Stability of the weighted CPIT CDF]
\label{thm:dkw_stability}
Suppose that, conditional on $\mathcal D_{\mathrm tr}$ and
$\mathcal D_{\mathrm bias}$, the calibration PIT values are independent with common
CDF
\[
C^\star(t)
=
\mathbb P\{u_i\le t\mid\mathcal D_{\mathrm tr},\mathcal D_{\mathrm bias}\},
\qquad t\in[0,1].
\]
Assume $C^\star\{m/(m+1)\}>0$, and define the oracle normalized rank-cell
weights by
\[
w_j^\star
=
\frac{C^\star\{j/(m+1)\}-C^\star\{(j-1)/(m+1)\}}
     {C^\star\{m/(m+1)\}},
\qquad j=1,\ldots,m.
\]
Let $F_{x,C^\star}^\tau$ be the smoothed CDF formed from these oracle weights,
using the same ordered adjusted samples and bandwidths as
$\widetilde F_x^\tau$. For $\varepsilon>0$, define
\[
\delta_{N,\varepsilon}=\varepsilon+\frac{1}{N+1}.
\]
If $\delta_{N,\varepsilon}<C^\star\{m/(m+1)\}$, then for every collection
$\mathcal X_0$ of covariate values,
\begin{equation}
\mathbb P\!\left\{
\sup_{x\in\mathcal X_0}\sup_{y\in\mathbb R}
\left|
\widetilde F_x^\tau(y)-F_{x,C^\star}^\tau(y)
\right|
>
\frac{2\delta_{N,\varepsilon}}
{C^\star\{m/(m+1)\}-\delta_{N,\varepsilon}}
\,\middle|\,
\mathcal D_{\mathrm tr},\mathcal D_{\mathrm bias}
\right\}
\le
2\exp(-2N\varepsilon^2).
\label{eq:dkw_stability}
\end{equation}
The same bound holds for the corresponding unsmoothed weighted empirical CDFs.
\end{theorem}

The theorem compares the fitted CPIT CDF with the oracle finite-support CDF
obtained from the population PIT calibration map $C^\star$. The term
$(N+1)^{-1}$ is the conformal rank correction in $\widehat C$. The bound is
uniform in the covariate values, adjusted sample locations, and bandwidths
because these quantities enter both CDFs identically, with only their common weights
differing. Together, Theorems~\ref{thm:bccpit} and \ref{thm:dkw_stability} separate
two roles of the calibration sample. Conformal ranking gives finite-sample
calibration in probability, while a one-dimensional empirical-process bound
controls estimation of the predictive CDF.

To distinguish this calibration-map error from other sources of approximation,
let
\[
F_x^0(y)=\mathbb P(Y\le y\mid X=x)
\]
be the true conditional CDF, let $H_x$ be the population CDF of the
bias-adjusted generator, and define the population oracle recalibration
\[
F_x^\star(y)=C^\star\{H_x(y)\}.
\]
For each $x$,
\[
\sup_y|\widetilde F_x^\tau(y)-F_x^0(y)|
\le
\sup_y|\widetilde F_x^\tau(y)-F_x^\star(y)|
+
\sup_y|F_x^\star(y)-F_x^0(y)|.
\]
The first term is controlled by Theorem~\ref{thm:dkw_stability}. The second term is an oracle approximation error. It vanishes only when the true conditional law can be represented as a global PIT recalibration of the bias-adjusted generator,
\[
F_x^0=C^\star\circ H_x.
\]
Thus CPIT provides finite-sample rank calibration and a stable predictive law,
whereas closeness to the full conditional distribution additionally depends on
the adequacy of the bias-adjusted generator and the use of a global calibration
map.

\section{Simulation Studies}
\label{sec:sim}

\subsection{Setup and metrics}
\label{sec:sim_setup}

We use three controlled simulations to separate the main ways in which a sample-only generator can fail. In each Monte Carlo replicate, the generator is specified analytically, so no training split is needed, and the source of misspecification is known. The data are split into bias, calibration, and test sets with
\[
 n_{\mathrm{bias}}=1000,
 \qquad
 n_{\mathrm{cal}}=N=1000,
 \qquad
 n_{\mathrm{test}}=5000.
\]
Unless otherwise stated, we use $m=100$ generator samples per covariate value and repeat each configuration over 100 replications. All post-processing methods receive only the generator samples and the observed responses.

The three designs are as follows.
\begin{enumerate}[leftmargin=1.6em]
\item Global location-scale distortion:
$X\sim \Unif(0,1)$,
$Y\mid X=x\sim N\{\mu(x),0.15^2\}$, and
$\mu(x)=\sin(2\pi x)$. The generator is globally biased:
\[
\widehat Y\mid X=x
\sim
N\!\left(
\frac{\mu(x)-0.25}{1.3},
\left(\frac{0.15}{1.3}\right)^2
\right).
\]
This design is favorable to a global affine correction.

\item Covariate-dependent location-scale distortion:
$X\sim \Unif(0,1)$,
$Y=\sin(2\pi X)+\sigma(X)\varepsilon$,
$\sigma(x)=0.05+0.25x$, and $\varepsilon\sim N(0,1)$. The generator has an
$x$-dependent mean and scale error:
\[
\widehat Y\mid X=x
\sim
N\!\Bigl(
0.85\sin(2\pi x)+0.15x-0.05,
\{0.08+0.10(1-x)\}^2
\Bigr).
\]
This design requires covariate-adaptive correction.

\item Distributional shape misspecification:
$X\sim \Unif(0,1)$,
$Y=\sin(2\pi X)+\sigma(X)(\eta-1)$,
$\sigma(x)=0.10+0.20x$, and $\eta\sim \mathrm{Exp}(1)$. The generator matches
the first two conditional moments but imposes a Gaussian predictive shape:
\[
\widehat Y\mid X=x
\sim
N\!\bigl(\sin(2\pi x),\sigma^2(x)\bigr).
\]
Thus affine correction alone cannot repair the skewed conditional law.
\end{enumerate}

We compare the raw generator, global and GAM affine bias corrections, CPIT
applied after each bias correction, and the fixed-level split-conformal
baselines in Section~\ref{sec:score_baselines}. We assess distributional
calibration using randomized PIT histograms, the Cram\'er--von Mises (CvM)
distance of the PIT distribution from uniformity, and quantile calibration
error (QErr). Specifically, QErr is the mean absolute difference between the
empirical coverage of each predictive quantile and its nominal level, averaged
over $\{0.05,0.10,\ldots,0.95\}$. We also report mean CRPS, central interval
coverage and length, and local diagnostics within five bins of $X$. Detailed
metric definitions, interval results for
$\alpha\in\{0.02,0.05,0.10,0.20\}$, local bin diagnostics, and the
Simulation~3 upper-tail QErr results are provided in the supplementary
material.

\subsection{Results}
\label{sec:sim_results}

Tables~\ref{tab:sim-dist-all} and~\ref{tab:sim-interval90} summarize the global
distributional diagnostics and the nominal 90\% interval results, respectively.
Figure~\ref{fig:sim-pit} shows the corresponding randomized PIT histograms,
averaged over the 100 Monte Carlo replications.

\begin{table}[tb]
\centering
\scriptsize
\setlength{\tabcolsep}{4pt}
\renewcommand{\arraystretch}{1.12}
\caption{Global distributional diagnostics for the three simulations over 100 replications. Entries are means, with Monte Carlo standard errors in parentheses. Smaller values indicate better performance for CvM, QErr, and mean CRPS.}
\label{tab:sim-dist-all}
\smallskip
\begin{tabular}{llrcc}
\toprule
Simulation & Method & CvM & QErr & Mean CRPS \\
\midrule
1 & Raw
& 474.5615 (1.3530) & 0.2745 (0.0004) & 0.1915 (0.0002) \\
  & BC(Global)
& 0.5884 (0.0701) & 0.0098 (0.0005) & 0.0854 (0.0001) \\
  & BC(GAM)
& 0.6671 (0.0729) & 0.0100 (0.0005) & 0.0856 (0.0001) \\
  & CPIT(Global)
& 0.9462 (0.0828) & 0.0114 (0.0005) & 0.0854 (0.0001) \\
  & CPIT(GAM)
& 0.9433 (0.0832) & 0.0111 (0.0005) & 0.0857 (0.0001) \\
\addlinespace
\midrule
2 & Raw
& 47.3461 (0.2167) & 0.0885 (0.0002) & 0.1436 (0.0002) \\
  & BC(Global)
& 3.8295 (0.1877) & 0.0243 (0.0007) & 0.1138 (0.0002) \\
  & BC(GAM)
& 0.6950 (0.0754) & 0.0101 (0.0005) & 0.1002 (0.0001) \\
  & CPIT(Global)
& 0.8816 (0.0743) & 0.0111 (0.0004) & 0.1135 (0.0002) \\
  & CPIT(GAM)
& 0.9550 (0.0810) & 0.0111 (0.0005) & 0.1002 (0.0001) \\
\addlinespace
\midrule
3 & Raw
& 35.6368 (0.2793) & 0.0757 (0.0003) & 0.1066 (0.0002) \\
  & BC(Global)
& 35.0022 (0.5079) & 0.0746 (0.0006) & 0.1066 (0.0002) \\
  & BC(GAM)
& 31.4727 (0.6454) & 0.0706 (0.0007) & 0.1069 (0.0002) \\
  & CPIT(Global)
& 1.0241 (0.1086) & 0.0120 (0.0005) & 0.1023 (0.0002) \\
  & CPIT(GAM)
& 1.0251 (0.1130) & 0.0120 (0.0006) & 0.1028 (0.0002) \\
\bottomrule
\end{tabular}
\end{table}

\begin{table}[tb]
\centering
\scriptsize
\setlength{\tabcolsep}{4pt}
\renewcommand{\arraystretch}{1.12}
\caption{Nominal 90\% interval coverage and average length for the three simulations. Each entry reports mean coverage followed by mean length, separated by a slash. Standard errors, additional nominal levels, and HDR variants are reported in the supplementary material.}
\label{tab:sim-interval90}
\smallskip
\begin{tabular}{lccc}
\toprule
Method & Simulation 1 & Simulation 2 & Simulation 3 \\
\midrule
Raw & 0.444 / 0.3683 & 0.641 / 0.4150 & 0.924 / 0.6384 \\
BC(Global) & 0.883 / 0.4794 & 0.801 / 0.5463 & 0.923 / 0.6386 \\
BC(GAM) & 0.885 / 0.4838 & 0.885 / 0.5668 & 0.917 / 0.6264 \\
CPIT(Global) & 0.890 / 0.4901 & 0.887 / 0.7735 & 0.888 / 0.5867 \\
CPIT(GAM) & 0.891 / 0.4930 & 0.891 / 0.5783 & 0.887 / 0.5956 \\
\addlinespace
\midrule
QR(Raw) & 0.900 / 0.9677 & 0.902 / 0.8571 & 0.901 / 0.5584 \\
QR(Global) & 0.901 / 0.5052 & 0.902 / 0.7592 & 0.901 / 0.5612 \\
QR(GAM) & 0.901 / 0.5076 & 0.902 / 0.5908 & 0.901 / 0.5763 \\
SR(Raw) & 0.899 / 0.9705 & 0.902 / 0.9839 & 0.900 / 0.5302 \\
SR(Global) & 0.901 / 0.5007 & 0.902 / 0.8164 & 0.900 / 0.5338 \\
SR(GAM) & 0.901 / 0.5033 & 0.902 / 0.5902 & 0.900 / 0.5516 \\
eCRPS(Raw) & 0.899 / 0.9635 & 0.902 / 0.8035 & 0.899 / 0.5361 \\
eCRPS(Global) & 0.900 / 0.4961 & 0.901 / 0.6816 & 0.899 / 0.5371 \\
eCRPS(GAM) & 0.900 / 0.4975 & 0.902 / 0.6002 & 0.899 / 0.5422 \\
\bottomrule
\end{tabular}
\end{table}

\begin{figure}[tb]
\centering
\maybeincludegraphics[width=0.8\textwidth]{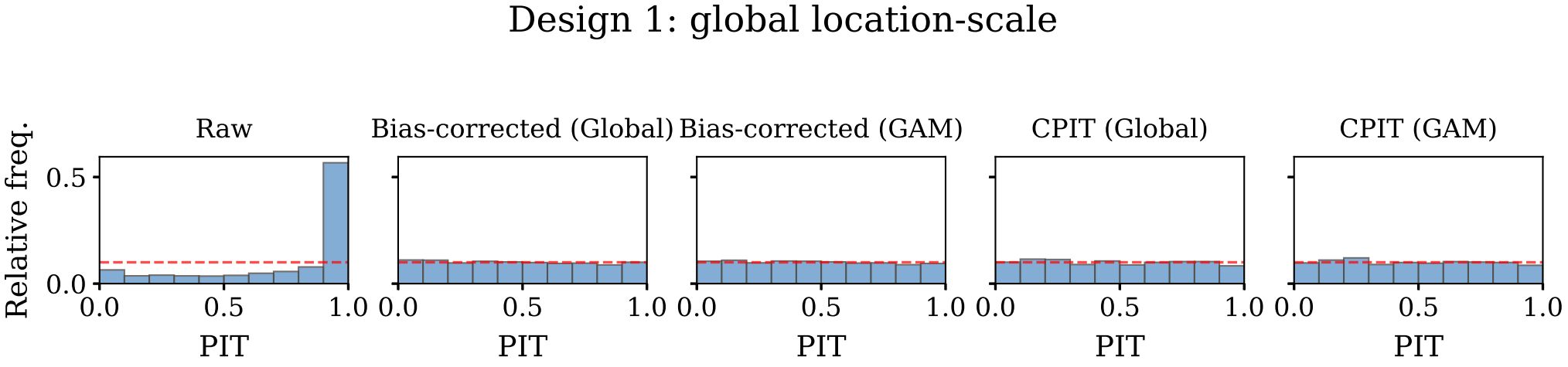}
\vspace{0.2em}\\
\maybeincludegraphics[width=0.8\textwidth]{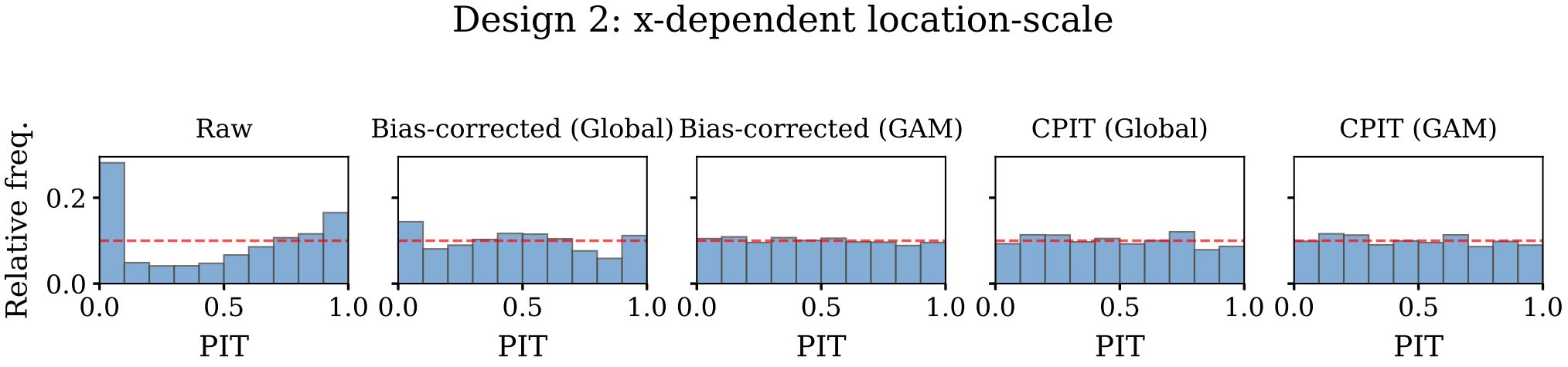}
\vspace{0.2em}\\
\maybeincludegraphics[width=0.8\textwidth]{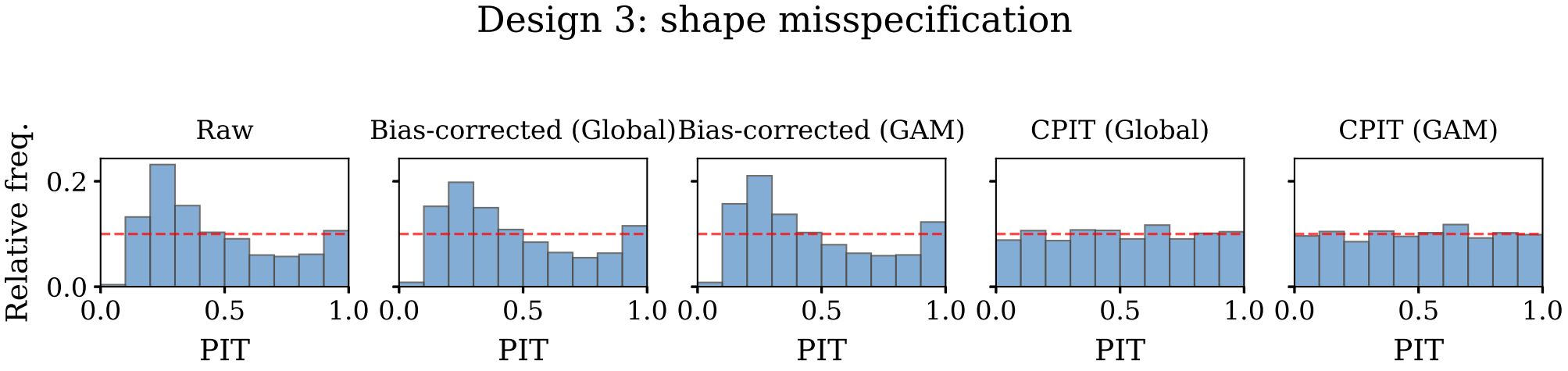}
\caption{Randomized PIT histograms for the three simulation designs, averaged over 100 replications. Within each row, the panels show the raw generator, global and GAM bias correction, and CPIT after each bias correction. The dashed line indicates the uniform reference.}
\label{fig:sim-pit}
\end{figure}

In Simulation~1, the global affine model is correctly specified. BC(Global)
reduces CvM from $474.5615$ to $0.5884$, QErr from $0.2745$ to $0.0098$, and
mean CRPS from $0.1915$ to $0.0854$. CPIT retains essentially the same CRPS
and yields nearly uniform PIT histograms. Its CvM values, $0.9462$ and
$0.9433$, are somewhat larger than the $0.5884$ attained by the correctly
specified global affine correction. As shown in Table~\ref{tab:sim-interval90},
CPIT also moves central 90\% coverage closer to the nominal level: coverage
increases from $0.883$ and $0.885$ under BC(Global) and BC(GAM) to $0.890$ and
$0.891$ under the corresponding CPIT corrections, with only modest increases
in average length.

Simulation~2 highlights the need for covariate-adaptive correction. BC(Global)
improves the raw generator but leaves appreciable miscalibration, with CvM
$3.8295$, QErr $0.0243$, and central 90\% coverage $0.801$. BC(GAM) lowers CvM
and QErr to $0.6950$ and $0.0101$, respectively, and raises coverage to
$0.885$. Both CPIT variants yield strong global calibration, with CvM below
$0.96$ and QErr $0.0111$. CPIT(GAM) additionally preserves the lower mean CRPS
of $0.1002$ and moves coverage to $0.891$, while maintaining a relatively small average length. The
supplementary binwise diagnostics reveal a distinction that is not apparent
from the global summaries: CPIT(Global) remains uneven across $X$, whereas
BC(GAM) and CPIT(GAM) achieve nearly uniform coverage across the five
covariate bins. Thus, favorable marginal diagnostics can mask residual
covariate-dependent miscalibration.

Simulation~3 isolates shape misspecification that affine location-scale
correction cannot remove. Bias correction alone has little effect: CvM remains
above $31$, QErr remains near $0.07$, and mean CRPS is essentially unchanged.
By contrast, CPIT reduces CvM to about $1.02$ and QErr to $0.0120$, while also
lowering mean CRPS to $0.1023$ under global correction and $0.1028$ under GAM
correction. Supplementary Table~S7 reports upper-tail calibration over quantile levels
$\tau\in\{0.96,0.97,0.98,0.99\}$. The corresponding upper-tail QErr decreases
from $0.0328$ for the raw generator and $0.0331$ after BC(Global) to $0.0134$
for CPIT(Global) and $0.0162$ for CPIT(GAM). The central CPIT intervals in
Table~\ref{tab:sim-interval90} have 90\% coverages $0.888$ and $0.887$, with
average lengths $0.5867$ and $0.5956$, respectively. The smoothed HDR variants
reported in the supplementary material have coverages $0.916$ and $0.915$.
Together, these results illustrate CPIT's role as a distributional
recalibrator rather than merely a location-scale adjustment.

Table~\ref{tab:sim-interval90} also compares CPIT with the interval-only QR, SR, and eCRPS split-conformal baselines. As expected, these methods attain coverage close to the nominal 90\% level in all three simulations. CPIT's central intervals have coverage between $0.887$ and $0.891$ and are competitive in length in Simulations~1 and~2, particularly after GAM correction in Simulation~2; the score-based baselines are in general slightly shorter in Simulation~3. The purpose of this comparison is therefore not to suggest that nominal fixed-level coverage is difficult to obtain, but to identify what CPIT adds beyond it. Unlike the interval-only baselines, CPIT returns a calibrated predictive CDF, so its quantiles, intervals, threshold probabilities, tail summaries, and resamples all arise from a single coherent predictive law.


\section{WeatherBench 2 Precipitation Applications}
\label{sec:weatherbench}

We evaluate CPIT using WeatherBench 2 forecasts of 24-hour accumulated
precipitation from the 50-member ECMWF Integrated Forecasting System
ensemble (IFS-ENS), with ERA5 reanalysis fields used for verification
\citep{rasp2024weatherbench2}. The forecasts are initialized twice daily,
at 00 and 12 UTC, and evaluated at a 24-hour lead over 2018-2022.
Forecasts and verifications are represented on the $1.5^\circ$
equiangular grid comprising 240 longitude points and 121 latitude points,
including both poles. The data are available through the WeatherBench 2
data archive.\footnote{\url{https://weatherbench2.readthedocs.io/en/latest/data-guide.html}}
After removing three initialization times for which either the IFS-ENS
forecast or the corresponding ERA5 verification is unavailable,
$n=3{,}649$ forecast-verification cases remain. We randomly partition
these cases into a bias-correction set
($n_{\mathrm{bias}}=1{,}204$), a PIT-calibration set
($n_{\mathrm{cal}}=N=1{,}204$), and a test set
($n_{\mathrm{test}}=1{,}241$).

For each forecast-verification case, we map the full spatial field to a scalar target. The
sample-only predictive distribution for that target is the empirical
distribution of the $m=50$ IFS-ENS ensemble members after the same spatial
mapping. Let $Z_i(s)$ denote the ERA5 verifying precipitation at grid cell
$s$, and let $\widehat Z_i^{(b)}(s)$, $b=1,\ldots,m$, denote the corresponding
IFS-ENS ensemble precipitation forecasts. For target functional $T_\ell$,
define
\[
Y_{i\ell}=T_\ell(Z_i),
\qquad
\widehat Y_{i\ell}^{(b)}=T_\ell(\widehat Z_i^{(b)}),
\qquad b=1,\ldots,m .
\]
Thus $Y_{i\ell}$ is the scalar verifying response, while
$\widehat Y_{i\ell}^{(1)},\ldots,\widehat Y_{i\ell}^{(m)}$ are the
sample-only predictive draws supplied by IFS-ENS for the same target and
verification time. The covariate $X_{i\ell}$ collects the information
available for post-processing at forecast initialization, including the date,
the ensemble forecast, and auxiliary atmospheric fields. For the GAM affine
correction, this information is summarized by a low-dimensional feature vector.

We consider two precipitation experiments. The Europe experiment uses the
region $35^\circ$-$75^\circ$N and $12.5^\circ$W-$42.5^\circ$E. It evaluates
both the area-weighted regional mean, $T_{\mathrm EU}^{\mathrm mean}$, and the
regional upper-tail target, $T_{\mathrm EU}^{0.95}$. The Taiwan experiment focuses
on four individual WeatherBench-2 grid cells covering western and eastern
Taiwan. Their domains are shown in Figure~\ref{fig:wb-eu-map}.

\begin{figure}[tb]
\centering
\begin{tabular}{cc}
(a) & (b) \\
\maybeincludegraphics[width=0.35\textwidth]{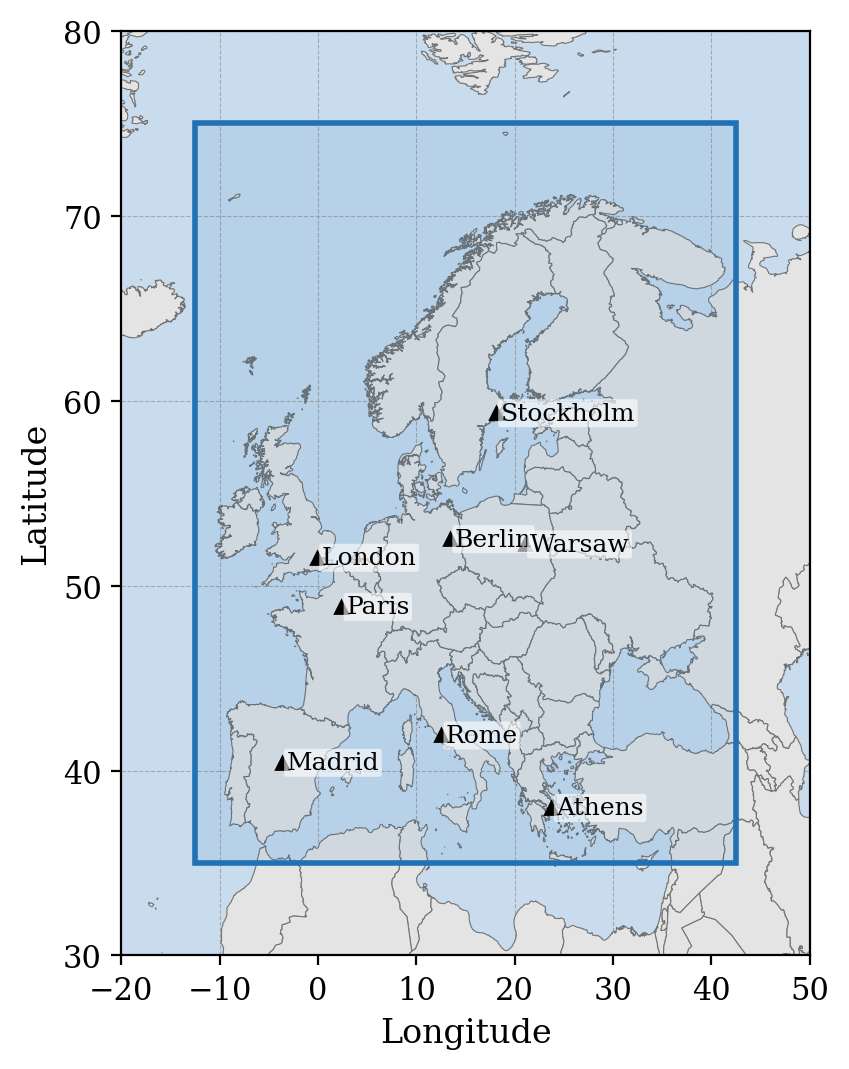} &
\maybeincludegraphics[width=0.37\textwidth]{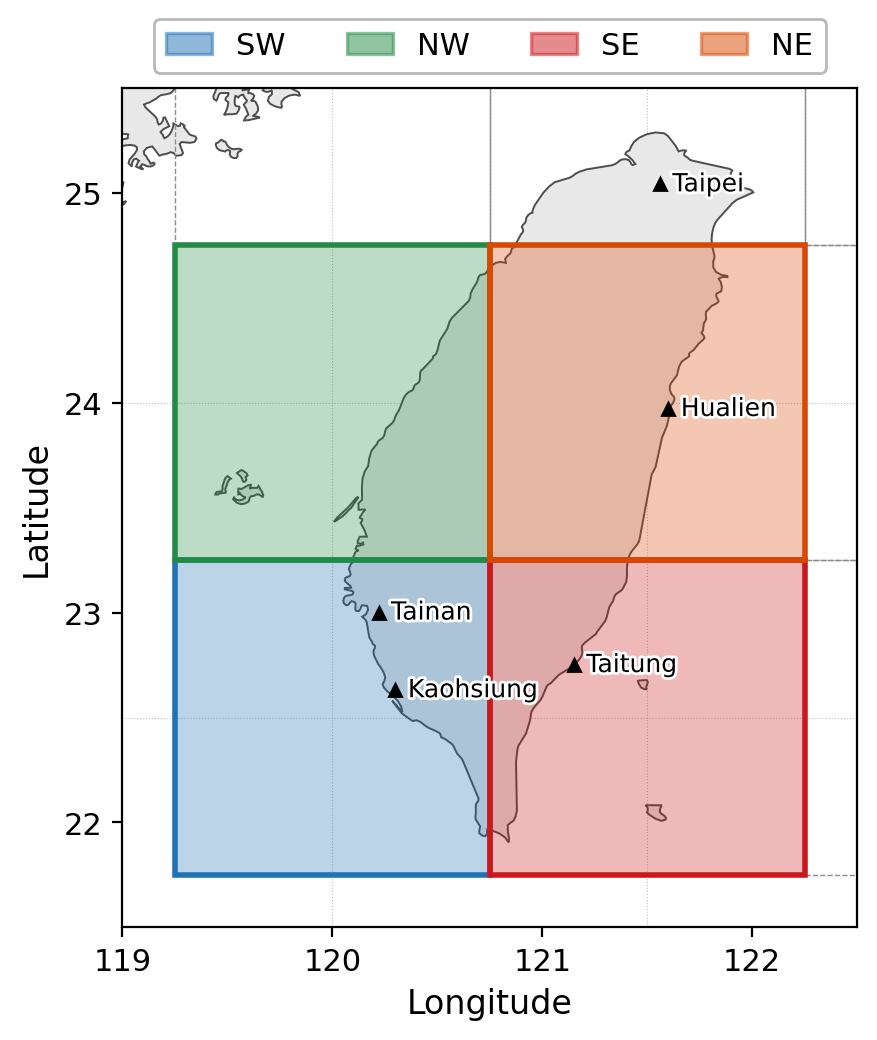}
\end{tabular}
\vspace{-0.3cm}
\caption{(a) Europe WeatherBench-2 region used for area-weighted aggregation:
$35^\circ$-$75^\circ$N and $-12.5^\circ$-$42.5^\circ$E. (b) Four WB2 $1.5^\circ$ grid cells used in the Taiwan experiment. The cells are SW $(22.5^\circ{\mathrm N},120^\circ{\mathrm E})$, NW $(24^\circ{\mathrm N},120^\circ{\mathrm E})$, SE $(22.5^\circ{\mathrm N},121.5^\circ{\mathrm E})$, and NE $(24^\circ{\mathrm N},121.5^\circ{\mathrm E})$.}
\label{fig:wb-eu-map}
\end{figure}

For a region $R$, let $a_s$ denote the area weight of grid cell $s$. The area-weighted mean target is
\[
T_R^{\mathrm mean}(Z_i)
=
\frac{\sum_{s\in R}a_s Z_i(s)}{\sum_{s\in R}a_s}.
\]
We also consider the weighted empirical upper-tail summary $T_R^q(Z_i)$,
defined as the weighted $q$-quantile of $\{Z_i(s):s\in R\}$ with weights
$\{a_s:s\in R\}$. For the Taiwan four-cell experiment, each target is a
singleton cell. If $s_\ell$ is one of the four selected Taiwan cells, then
$T_\ell(Z_i)=Z_i(s_\ell)$.

For both Europe and Taiwan, precipitation is retained in meters for CPIT calibration and evaluation, while the affine correction is fitted on the log-millimeter scale
\[
g(y)=\log(y+0.1),
\]
where $y$ is precipitation in millimeters. On the transformed scale, define
\[
\widehat G_{i\ell}^{(b)}=g\big(\widehat Y_{i\ell}^{(b)}\big),
\qquad
\widehat\mu_{i\ell}=\frac{1}{m}\sum_{b=1}^m \widehat G_{i\ell}^{(b)},
\]
and
\[
\widehat\sigma_{i\ell}^2
=
\frac{1}{m-1}\sum_{b=1}^m\big(\widehat G_{i\ell}^{(b)}-\widehat\mu_{i\ell}\big)^2 .
\]
For Europe, the GAM affine correction uses the three-feature vector
\[
\phi^{\mathrm EU}_{i\ell}
=
\bigl(d_i,\widehat\mu_{i\ell},\log(\widehat\sigma_{i\ell}+\varepsilon)\bigr)' ,
\]
whereas for Taiwan it uses the six-feature vector
\[
\phi^{\mathrm TW}_{i\ell}
=
\bigl(
d_i,\widehat\mu_{i\ell},\log(\widehat\sigma_{i\ell}+\varepsilon),
u_{850,i\ell},v_{850,i\ell},{\mathrm TCWV}_{i\ell}
\bigr)' .
\]
Here $\varepsilon>0$ is a small numerical constant. In both analyses, $d_i$ denotes day of year and is modeled with a cyclic P-spline to capture annual periodicity. Each smooth term uses five spline basis functions, and the smoothing penalty is selected by GCV. The variables
$u_{850,i\ell}$ and $v_{850,i\ell}$ are the local zonal and meridional wind
components at 850 hPa, and ${\mathrm TCWV}_{i\ell}$ denotes total column water
vapor.

We assess threshold-event performance using Brier skill scores. In the formulas below, $c$ is expressed on the original millimeter scale. For threshold $c$, let
\[
\widehat\pi_{i\ell}(c)
=
1-\widetilde F^\tau_{i\ell}(c)
\]
denote the CPIT exceedance probability on the original precipitation scale,
where $\widetilde F^\tau_{i\ell}$ is the smoothed weighted CDF in
\eqref{eq:smoothed-cdf}. The Brier score is
\[
\mathrm{BS}_\ell(c)
=
\frac1{n_{\mathrm test}}
\sum_{i\in\mathcal I_{\mathrm test}}
\left\{
\widehat\pi_{i\ell}(c)-\mathbf 1(Y_{i\ell}>c)
\right\}^2,
\]
and the Brier skill score is
\[
\mathrm{BSS}_\ell(c)
=
1-
\frac{\mathrm{BS}_\ell(c)}{\mathrm{BS}_{\ell,\rm ref}(c)}.
\]
Here $\mathrm{BS}_{\ell,\rm ref}(c)$ is the Brier score of the climatological
reference forecast, using the empirical exceedance rate from the combined bias
and calibration splits. For Taiwan, we consider $c\in\{10,20,40,80\}$. The 80 mm event
occurs only about five times in the $1241$ Taiwan test cases, so BSS at this
threshold is statistically unstable and can become negative after only a small
number of false alarms.



\subsection{Europe regional targets}

The European experiment evaluates CPIT for two regional precipitation
functionals with distinct meteorological interpretations. The domain shown in
Figure~\ref{fig:wb-eu-map}(a), spanning $35^\circ$-$75^\circ$N and
$12.5^\circ$W-$42.5^\circ$E, covers the North Atlantic-European storm
track, the Mediterranean basin, and northern Europe. All spatial summaries
use area weights, so that each grid cell contributes according to its physical
area rather than receiving equal weight on the latitude-longitude grid. The
regional mean summarizes the domain-wide 24-hour precipitation burden. The
p95 target, defined as the area-weighted 95th percentile of 24-hour
precipitation across grid cells, summarizes the spatial upper tail and is
therefore more sensitive to localized precipitation maxima associated with
frontal, convective, or orographic processes.

Table~\ref{tab:wb-eu-dist} reports complementary diagnostics of distributional
reliability, sharpness, and threshold-event skill. We use randomized PIT CvM
and QErr to assess full-distribution and quantile calibration, empirical
coverage of central 90\% intervals to assess interval validity, mean CRPS to
summarize proper-score performance, and Brier skill scores relative to
empirical climatology for exceedance probabilities.

\begin{table}[tb]
\centering
\scriptsize
\caption{WeatherBench-2 Europe results for various methods.}
\label{tab:wb-eu-dist}
\smallskip
\resizebox{\textwidth}{!}{%
\begin{tabular}{llcccccccc}
\toprule
Target & Method & CvM & QErr & $90\%$ Cov & Mean CRPS & BSS$(>2.5)$ & BSS$(>3.0)$ & BSS$(>3.5)$ & BSS$(>4.0)$ \\
\midrule
Mean & Raw ensemble  & 2.8318 & 0.0344 & 0.929 & 0.05 & 0.892 & 0.904 & 0.895 & 0.886 \\
Mean & BC (Global)   & 2.9386 & 0.0374 & 0.932 & 0.05 & --    & --    & --    & --    \\
Mean & BC (GAM)      & 0.3303 & 0.0164 & 0.867 & 0.04 & --    & --    & --    & --    \\
Mean & CPIT (Raw)    & 0.8739 & 0.0250 & 0.892 & 0.04 & 0.895 & 0.904 & 0.903 & 0.891 \\
Mean & CPIT (Global) & 0.7056 & 0.0225 & 0.879 & 0.04 & 0.895 & 0.908 & 0.911 & 0.890 \\
Mean & CPIT (GAM)    & 0.9078 & 0.0260 & 0.882 & 0.04 & 0.900 & 0.914 & 0.918 & 0.885 \\
\hline
p95 & Raw ensemble  & 4.7708 & 0.0524 & 0.910 & 0.31 & 1.000 & 1.000 & 0.932 & 0.482 \\
p95 & BC (Global)   & 1.6337 & 0.0259 & 0.911 & 0.30 & --    & --    & --    & --    \\
p95 & BC (GAM)      & 0.1232 & 0.0090 & 0.856 & 0.29 & --    & --    & --    & --    \\
p95 & CPIT (Raw)    & 0.7052 & 0.0204 & 0.898 & 0.30 & 1.000 & 1.000 & 0.952 & 0.417 \\
p95 & CPIT (Global) & 0.5433 & 0.0195 & 0.886 & 0.30 & 1.000 & 1.000 & 0.965 & 0.425 \\
p95 & CPIT (GAM)    & 0.4660 & 0.0188 & 0.870 & 0.29 & 1.000 & 1.000 & 0.911 & 0.529 \\
\bottomrule
\end{tabular}}
\end{table}

For the Europe mean target, the raw ensemble is slightly conservative at the
central 90\% level, with coverage $0.929$, but its CvM statistic is $2.8318$.
The global affine correction does not improve this behavior, whereas BC(GAM)
reduces CvM to $0.3303$ and QErr to $0.0164$ at the cost of undercoverage,
$0.867$. The unsmoothed CPIT variants reduce CvM to between $0.7056$ and
$0.9078$ and QErr to between $0.0225$ and $0.0260$. Their central coverages,
$0.879$-$0.892$, are slightly below nominal, but smoothing raises them to
$0.919$-$0.927$ in Table~\ref{tab:wb-eu-interval90}. CPIT(Global) has the
smallest CvM and QErr among the CPIT variants, while CPIT(GAM) has the strongest
Brier skill at the 2.5, 3.0, and 3.5 mm thresholds.

The p95 target shows stronger distributional error in the raw ensemble. Its
90\% coverage is $0.910$, yet its CvM statistic is $4.7708$. BC(GAM) produces
very small PIT and quantile errors, CvM $0.1232$ and QErr $0.0090$, but
undercovers at $0.856$. CPIT reduces CvM to $0.7052$, $0.5433$, and $0.4660$
for the raw, global, and GAM variants, respectively, while their unsmoothed
central coverages are $0.898$, $0.886$, and $0.870$. The smoothed variants
raise these coverages to $0.919$, $0.920$, and $0.904$. CPIT(GAM) has the
smallest p95 CvM and QErr and the largest BSS at the highest displayed
threshold. The lower p95 thresholds are nearly saturated in this split, so
their BSS values should be interpreted cautiously.

\begin{table}[tb]
\centering
\scriptsize
\caption{WeatherBench-2 Europe 90\% interval coverage and average length for various methods. Each entry is ``coverage / length", with length in millimeters. HDR denotes the shortest weighted interval from the calibrated predictive distribution. Smooth indicates using the kernel-smoothed CPIT distribution.}
\label{tab:wb-eu-interval90}
\smallskip
\begin{tabular}{lcc}
\toprule
Method & Mean & p95 \\
\midrule
Raw
& 0.929 / 0.32
& 0.910 / 1.94 \\
BC (GAM)
& 0.867 / 0.23
& 0.856 / 1.56 \\
CPIT (Raw)
& 0.892 / 0.28
& 0.898 / 1.87 \\
CPIT (Raw, smooth)
& 0.925 / 0.31
& 0.919 / 2.01 \\
CPIT (Raw, HDR)
& 0.874 / 0.26
& 0.862 / 1.66 \\
CPIT (Raw, HDR, smooth)
& 0.927 / 0.31
& 0.922 / 1.98 \\
CPIT (Global)
& 0.879 / 0.26
& 0.886 / 1.76 \\
CPIT (Global, smooth)
& 0.927 / 0.31
& 0.920 / 1.99 \\
CPIT (Global, HDR)
& 0.864 / 0.25
& 0.861 / 1.62 \\
CPIT (Global, HDR, smooth)
& 0.928 / 0.31
& 0.922 / 1.96 \\
CPIT (GAM)
& 0.882 / 0.24
& 0.870 / 1.64 \\
CPIT (GAM, smooth)
& 0.919 / 0.26
& 0.904 / 1.76 \\
CPIT (GAM, HDR)
& 0.851 / 0.22
& 0.836 / 1.47 \\
CPIT (GAM, HDR, smooth)
& 0.915 / 0.26
& 0.911 / 1.74 \\
QR (GAM)
& 0.905 / 0.26
& 0.886 / 1.68 \\
SR (GAM)
& 0.904 / 0.25
& 0.885 / 1.62 \\
eCRPS (GAM)
& 0.896 / 0.24
& 0.870 / 1.60 \\
\bottomrule
\end{tabular}
\end{table}

Table~\ref{tab:wb-eu-interval90} shows the same validity-sharpness tradeoff for
both targets. Unsmoothed central and HDR intervals are generally shorter than
the raw intervals but tend to undercover. Kernel smoothing increases coverage
with a moderate increase in length. For example, the p95 CPIT(Global) interval
changes from $0.886/1.76$ without smoothing to $0.920/1.99$ with smoothing,
whereas the p95 CPIT(GAM) interval changes from $0.870/1.64$ to $0.904/1.76$.
Thus the updated results support using the unsmoothed law for compact empirical
summaries and the smoothed law when interval calibration is the primary goal.



\subsection{Taiwan four-cell experiment}

Figure~\ref{fig:wb-eu-map}(b) shows the four WeatherBench-2 grid cells used in
the Taiwan experiment. This setting is more local and meteorologically more
demanding than the Europe regional-average experiment. Each target is a single
$1.5^\circ$ grid cell, so spatial averaging does not smooth out displacement
error, land-sea contrast, or unresolved orographic effects. The four cells
provide coarse proxies for southwestern Taiwan (SW), northwestern Taiwan (NW),
southeastern Taiwan (SE), and northeastern Taiwan (NE).

Table~\ref{tab:wb-tw-cell-clim} reports the observed exceedance frequencies in
the test split. The NE cell is the wettest at the 10 mm threshold, with
$P(Y>10{\mathrm mm})=0.276$. The 80 mm event is rare in all four cells, with
frequencies between $0.003$ and $0.006$, corresponding to only about 4-7
events per cell in $n_{\mathrm test}=1241$ cases. The 80 mm BSS values should
therefore be read as sensitivity diagnostics rather than stable estimates of
operational tail skill.

\begin{table}[tb]
\centering
\scriptsize
\caption{Taiwan four-cell experiment: cell locations and observed exceedance frequencies in the test split.}
\label{tab:wb-tw-cell-clim}
\smallskip
\begin{tabular}{llcccc}
\toprule
Cell & Location & $P(Y>10{\mathrm mm})$ & $P(Y>20{\mathrm mm})$ & $P(Y>40{\mathrm mm})$ & $P(Y>80{\mathrm mm})$ \\
\midrule
SW & $(22.5^\circ{\mathrm N},120^\circ{\mathrm E})$     & 0.128 & 0.061 & 0.022 & 0.005 \\
NW & $(24^\circ{\mathrm N},120^\circ{\mathrm E})$       & 0.115 & 0.053 & 0.015 & 0.003 \\
SE & $(22.5^\circ{\mathrm N},121.5^\circ{\mathrm E})$   & 0.147 & 0.058 & 0.019 & 0.006 \\
NE & $(24^\circ{\mathrm N},121.5^\circ{\mathrm E})$     & 0.276 & 0.091 & 0.023 & 0.005 \\
\bottomrule
\end{tabular}
\end{table}

Table~\ref{tab:wb-tw-fourcell-dist} reports distributional calibration and
probabilistic forecasting scores. The raw ensemble is substantially
underdispersed in SW, NW, and NE, with central 90\% coverages $0.763$, $0.737$,
and $0.734$. Their CvM statistics are $12.6724$, $7.4014$, and $5.3007$,
respectively. The SE cell is less severely miscalibrated but still undercovers
at $0.820$.

\begin{table}[tb]
\centering
\scriptsize
\caption{WeatherBench-2 Taiwan four-cell results for various methods.}
\label{tab:wb-tw-fourcell-dist}
\medskip
\resizebox{\textwidth}{!}{%
\begin{tabular}{llcccccccc}
\toprule
Cell & Method & CvM & QErr & $90\%$ Cov & Mean CRPS & BSS$(>10)$ & BSS$(>20)$ & BSS$(>40)$ & BSS$(>80)$ \\
\midrule
SW & Raw ensemble  & 12.6724 & 0.0976 & 0.763 & 0.95 & 0.707 & 0.710 & 0.555 & 0.017 \\
SW & BC (Global)   &  1.0747 & 0.0334 & 0.800 & 0.94 & --    & --    & --    & --    \\
SW & BC (GAM)      &  0.2163 & 0.0084 & 0.871 & 0.94 & --    & --    & --    & --    \\
SW & CPIT (Raw)    &  0.0900 & 0.0114 & 0.862 & 0.95 & 0.705 & 0.703 & 0.527 & 0.124 \\
SW & CPIT (Global) &  0.1482 & 0.0101 & 0.857 & 0.94 & 0.709 & 0.702 & 0.536 & 0.106 \\
SW & CPIT (GAM)    &  0.1679 & 0.0137 & 0.870 & 0.94 & 0.706 & 0.720 & 0.540 & 0.034 \\
\hline
NW & Raw ensemble  &  7.4014 & 0.0761 & 0.737 & 0.85 & 0.695 & 0.715 & 0.601 & 0.715 \\
NW & BC (Global)   & 10.2603 & 0.0861 & 0.712 & 0.85 & --    & --    & --    & --    \\
NW & BC (GAM)      &  0.3066 & 0.0159 & 0.876 & 0.82 & --    & --    & --    & --    \\
NW & CPIT (Raw)    &  0.7067 & 0.0249 & 0.877 & 0.84 & 0.707 & 0.706 & 0.601 & 0.694 \\
NW & CPIT (Global) &  0.6024 & 0.0244 & 0.853 & 0.85 & 0.709 & 0.704 & 0.590 & 0.680 \\
NW & CPIT (GAM)    &  0.2306 & 0.0130 & 0.890 & 0.82 & 0.713 & 0.727 & 0.605 & 0.756 \\
\hline
SE & Raw ensemble  &  0.4966 & 0.0253 & 0.820 & 1.14 & 0.687 & 0.583 & 0.533 & -0.370 \\
SE & BC (Global)   &  0.2791 & 0.0216 & 0.823 & 1.13 & --    & --    & --    & --    \\
SE & BC (GAM)      &  0.2145 & 0.0102 & 0.879 & 1.10 & --    & --    & --    & --    \\
SE & CPIT (Raw)    &  0.1770 & 0.0128 & 0.873 & 1.13 & 0.685 & 0.588 & 0.538 & -0.390 \\
SE & CPIT (Global) &  0.1749 & 0.0111 & 0.874 & 1.14 & 0.684 & 0.588 & 0.543 & -0.561 \\
SE & CPIT (GAM)    &  0.1139 & 0.0092 & 0.889 & 1.10 & 0.697 & 0.617 & 0.543 & -0.915 \\
\hline
NE & Raw ensemble  &  5.3007 & 0.0654 & 0.734 & 1.75 & 0.513 & 0.491 & 0.498 & 0.620 \\
NE & BC (Global)   &  4.5055 & 0.0639 & 0.731 & 1.78 & --    & --    & --    & --    \\
NE & BC (GAM)      &  0.1711 & 0.0179 & 0.839 & 1.66 & --    & --    & --    & --    \\
NE & CPIT (Raw)    &  0.3065 & 0.0214 & 0.827 & 1.75 & 0.528 & 0.482 & 0.467 & 0.623 \\
NE & CPIT (Global) &  0.2593 & 0.0181 & 0.856 & 1.75 & 0.537 & 0.478 & 0.414 & 0.612 \\
NE & CPIT (GAM)    &  0.1044 & 0.0120 & 0.895 & 1.66 & 0.580 & 0.492 & 0.503 & 0.324 \\
\bottomrule
\end{tabular}}
\end{table}

The global affine correction is unreliable in this local setting. It improves
SW, SE, and NE to varying degrees but worsens NW, where CvM increases from
$7.4014$ to $10.2603$. By contrast, BC(GAM) has average CvM $0.2271$, compared
with $6.4678$ for the raw ensemble, and average QErr $0.0131$, compared with
$0.0661$. Its average 90\% coverage is $0.8663$, so a flexible location-scale
correction alone still does not fully calibrate predictive uncertainty.

Averaged over the four cells, the CvM statistics for CPIT(Raw), CPIT(Global),
and CPIT(GAM) are $0.3201$, $0.2962$, and $0.1542$, respectively. CPIT(GAM)
therefore reduces average CvM by about 98\% relative to the raw ensemble. Its
average QErr is $0.0120$. The average unsmoothed central coverage improves from
$0.7635$ for the raw ensemble to $0.8860$ for CPIT(GAM), with NW and NE
improving to $0.890$ and $0.895$. The smoothed CPIT(GAM) intervals in
Table~\ref{tab:wb-tw-fourcell-interval90} raise average coverage further to
$0.9275$.

For the more frequent 10 and 20 mm events, CPIT(GAM) improves average BSS from
$0.6505$ to $0.6740$ and from $0.6248$ to $0.6390$, respectively. At 40 mm the
average BSS is essentially unchanged. At 80 mm, a few false alarms or misses
produce large cell-to-cell changes, including strongly negative values in SE,
so this threshold is not used to rank methods.

\begin{table}[tb]
\centering
\scriptsize
\caption{WeatherBench-2 Taiwan 90\% interval coverage and average length for the log-millimeter analysis. Each entry is coverage / length, with length in millimeters. HDR denotes the shortest weighted interval from the calibrated predictive distribution. Smooth intervals use the kernel-smoothed CPIT distribution.}
\label{tab:wb-tw-fourcell-interval90}
\smallskip
\begin{tabular}{lcccc}
\toprule
Method & SW & NW & SE & NE \\
\midrule
Raw
& 0.763 / 4.23
& 0.737 / 3.78
& 0.820 / 5.57
& 0.734 / 6.65 \\
BC (GAM)
& 0.871 / 5.28
& 0.876 / 4.57
& 0.879 / 6.34
& 0.839 / 8.57 \\
CPIT (Raw)
& 0.862 / 4.88
& 0.877 / 5.87
& 0.873 / 6.88
& 0.827 / 9.64 \\
CPIT (Raw, smooth)
& 0.919 / 5.48
& 0.922 / 5.83
& 0.898 / 7.36
& 0.857 / 10.18 \\
CPIT (Raw, HDR)
& 0.814 / 3.96
& 0.828 / 4.35
& 0.834 / 5.44
& 0.794 / 8.23 \\
CPIT (Raw, HDR, smooth)
& 0.910 / 5.09
& 0.912 / 5.25
& 0.892 / 6.82
& 0.865 / 9.76 \\
CPIT (Global)
& 0.857 / 4.97
& 0.853 / 5.82
& 0.874 / 6.93
& 0.856 / 11.02 \\
CPIT (Global, smooth)
& 0.918 / 5.68
& 0.915 / 5.83
& 0.898 / 7.39
& 0.874 / 11.07 \\
CPIT (Global, HDR)
& 0.836 / 4.26
& 0.802 / 4.32
& 0.837 / 5.71
& 0.810 / 8.40 \\
CPIT (Global, HDR, smooth)
& 0.908 / 5.22
& 0.907 / 5.25
& 0.889 / 6.88
& 0.873 / 10.25 \\
CPIT (GAM)
& 0.870 / 5.43
& 0.890 / 4.84
& 0.889 / 6.68
& 0.895 / 10.46 \\
CPIT (GAM, smooth)
& 0.923 / 6.13
& 0.941 / 5.48
& 0.929 / 7.48
& 0.917 / 10.64 \\
CPIT (GAM, HDR)
& 0.844 / 4.51
& 0.857 / 4.06
& 0.863 / 5.89
& 0.844 / 8.26 \\
CPIT (GAM, HDR, smooth)
& 0.913 / 5.76
& 0.930 / 5.17
& 0.921 / 7.11
& 0.903 / 10.01 \\
QR (Raw)
& 0.891 / 4.65
& 0.915 / 4.23
& 0.892 / 6.15
& 0.882 / 9.08 \\
QR (GAM)
& 0.882 / 5.32
& 0.922 / 4.68
& 0.911 / 6.57
& 0.895 / 9.29 \\
SR (GAM)
& 0.887 / 5.25
& 0.929 / 4.84
& 0.910 / 6.67
& 0.887 / 9.07 \\
eCRPS (GAM)
& 0.915 / 5.82
& 0.926 / 6.17
& 0.915 / 6.72
& 0.885 / 9.39 \\
\bottomrule
\end{tabular}
\end{table}

The raw intervals have average coverage $0.7635$ and average length $5.06$ mm.
For CPIT(GAM), the unsmoothed central intervals have average coverage $0.8860$
and length $6.85$ mm, and smoothing increases these to $0.9275$ and $7.43$ mm.
The unsmoothed HDR intervals are shorter, with average length $5.68$ mm, but
undercover at $0.8520$. Smoothed HDR intervals provide a more balanced
compromise, with average coverage $0.9168$ and length $7.01$ mm. These results
show that smoothing is particularly useful when the finite 50-member ensemble
limits the resolution of high-coverage intervals.

The QR, SR, and eCRPS baselines perform well under the fixed 90\% interval
criterion because each is calibrated specifically at that nominal level.
Their outputs, however, are level-specific intervals rather than reusable
predictive distributions. CPIT instead produces a single calibrated CDF from
which quantiles, intervals, and exceedance probabilities at arbitrary
thresholds can be derived coherently. For example, letting
$\widehat\pi_i(c)=1-\widehat F_i(c)$ denote the predicted probability that
precipitation exceeds $c$ in forecast case $i$, the resulting probabilities
are automatically monotone in the threshold:
\[
\widehat\pi_i(10\,\mathrm{mm})
\geq
\widehat\pi_i(20\,\mathrm{mm})
\geq
\widehat\pi_i(40\,\mathrm{mm})
\geq
\widehat\pi_i(80\,\mathrm{mm}).
\]

\subsection{Diagnostic plots}

Figures~\ref{fig:wb-eu-chat} and \ref{fig:wb-tw-chat} show the fitted calibration maps for selected test cases. The Europe GAM-adjusted profiles are close to the diagonal, while the Taiwan profiles show that the remaining departures after global correction vary by cell and are substantially reduced by the GAM correction. Figures~\ref{fig:wb-eu-pit} and \ref{fig:wb-tw-pit} show the updated randomized PIT histograms. These plots agree with Tables~\ref{tab:wb-eu-dist} and \ref{tab:wb-tw-fourcell-dist}. Specifically, bias correction removes much of the location-scale error, and CPIT flattens the remaining rank distortions.

\begin{figure}[tb]
\centering
\maybeincludegraphics[width=0.8\textwidth]{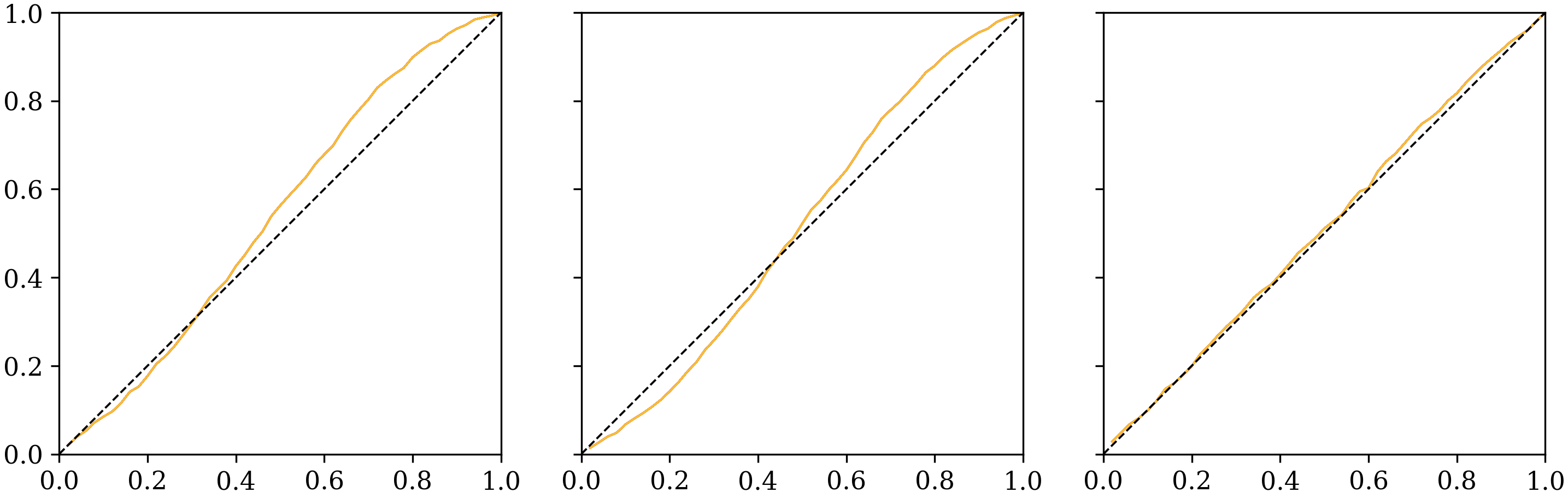}
\vspace{-0.2cm}
\caption{Europe empirical CPIT calibration profiles for the seed-123
log-millimeter mean-target analysis. Left to right: CPIT(Raw), CPIT(Global),
and CPIT(GAM). The dashed line is the uniform reference \(\widehat C(t)=t\).}
\label{fig:wb-eu-chat}
\end{figure}

\begin{figure}[tb]
\centering
\maybeincludegraphics[width=0.8\textwidth]{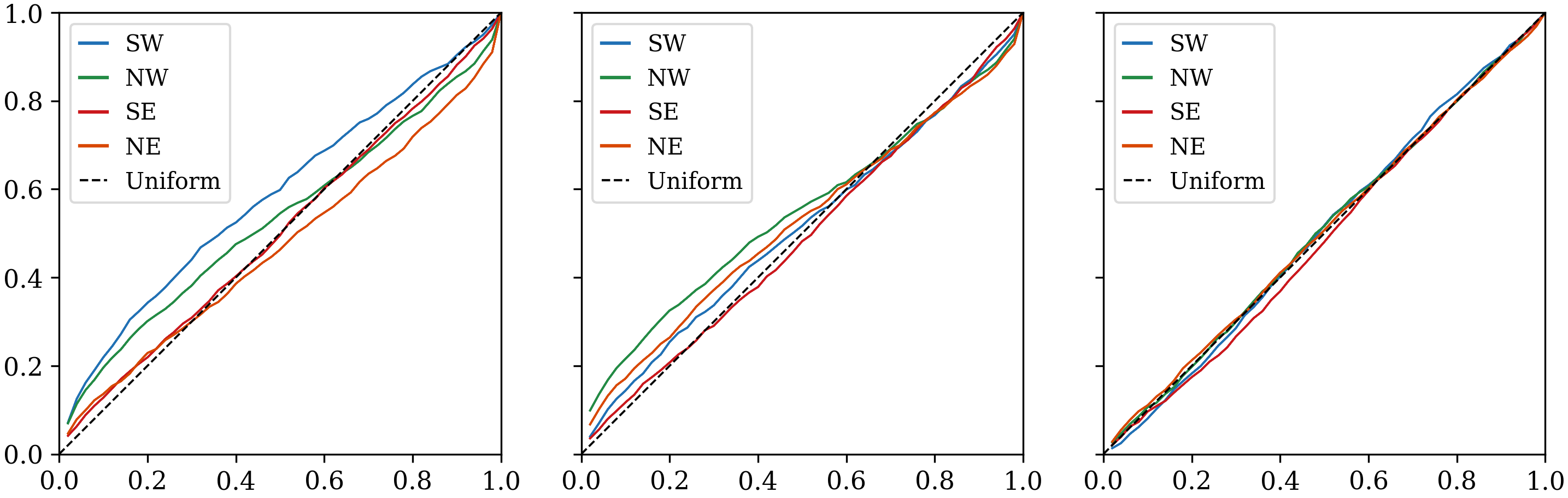}
\vspace{-0.2cm}
\caption{Taiwan empirical CPIT calibration profiles for the seed-123
log-millimeter mean-target analysis. Left to right: CPIT(Raw), CPIT(Global),
and CPIT(GAM). Colours distinguish the SW, NW, SE, and NE cells.}
\label{fig:wb-tw-chat}
\end{figure}

\begin{figure}[tb]
\centering
~\hspace{1cm}\footnotesize{Raw}\hspace{1.7cm}\footnotesize{BC}(Global)\hspace{1.3cm}\footnotesize{BC(GAM)}\hspace{1cm}\footnotesize{CPIT(Global)}\hspace{0.9cm}\footnotesize{CPIT(GAM)}\\
\maybeincludegraphics[width=0.8\textwidth]{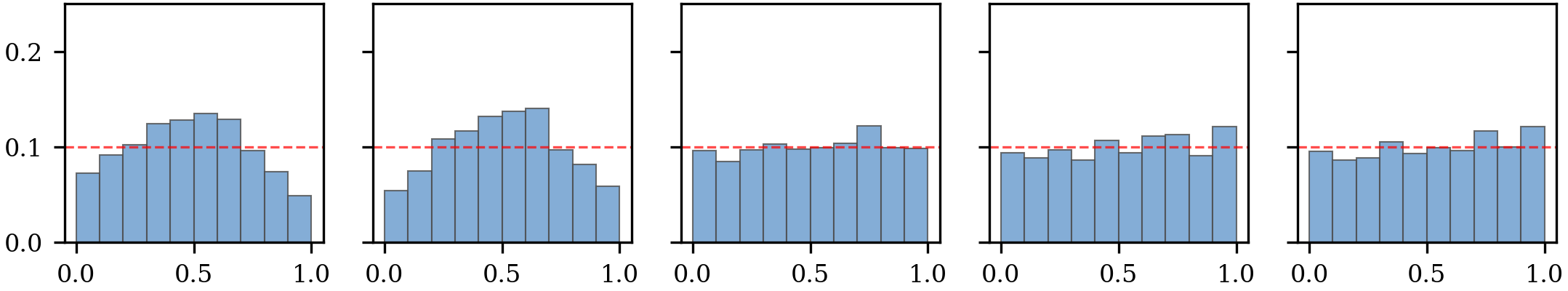}
\vspace{0.2em}\\
\maybeincludegraphics[width=0.8\textwidth]{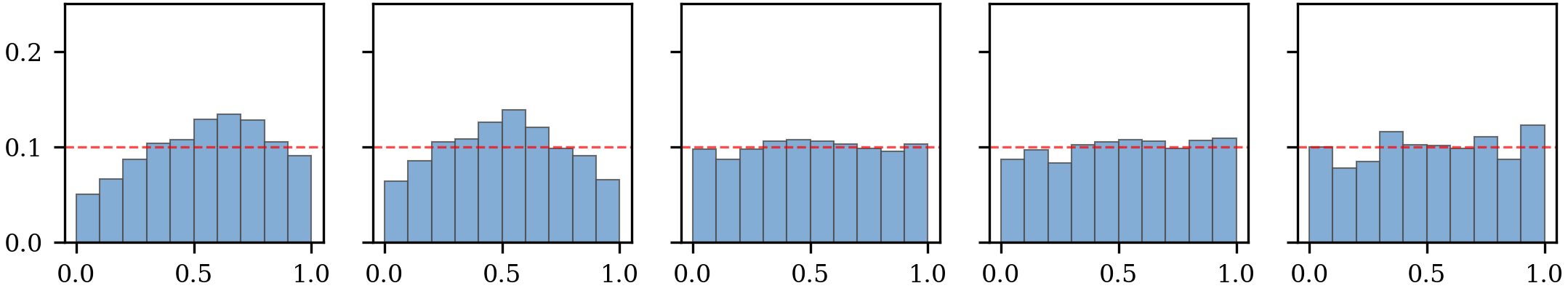}
\vspace{-0.3cm}\\
\caption{Europe randomized PIT histograms for an illustrative split. Top:
area-weighted mean target. Bottom: regional p95 target. Each row compares the
raw ensemble, bias-corrected distributions, and CPIT-calibrated distributions.}
\label{fig:wb-eu-pit}
\end{figure}

\begin{figure}[tb]
\centering
\begin{tabular}{cc}
\footnotesize SW & \footnotesize NW \\
\maybeincludegraphics[width=0.47\textwidth]{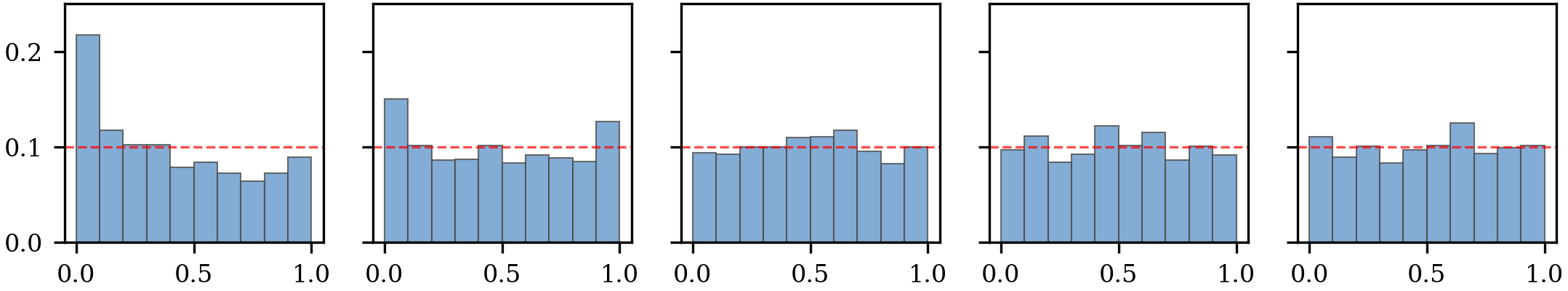} &
\maybeincludegraphics[width=0.47\textwidth]{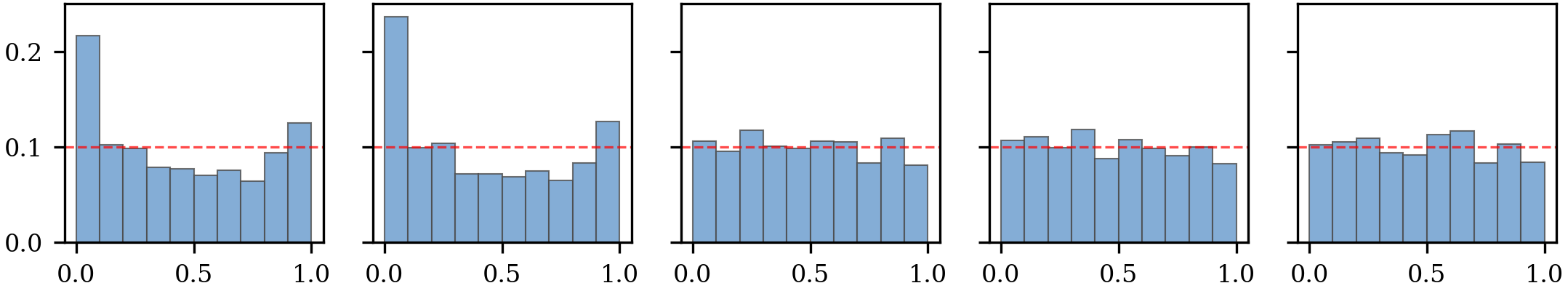} \\[0.3em]
\footnotesize SE & \footnotesize NE \\
\maybeincludegraphics[width=0.47\textwidth]{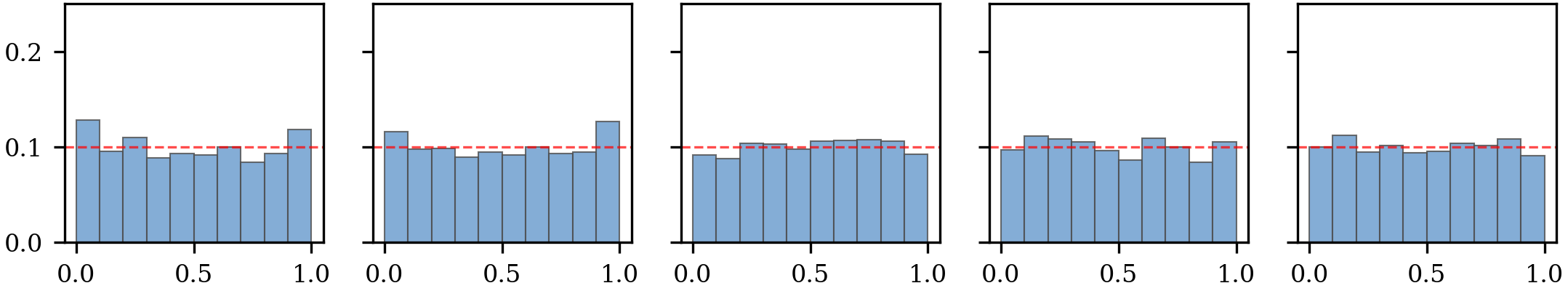} &
\maybeincludegraphics[width=0.47\textwidth]{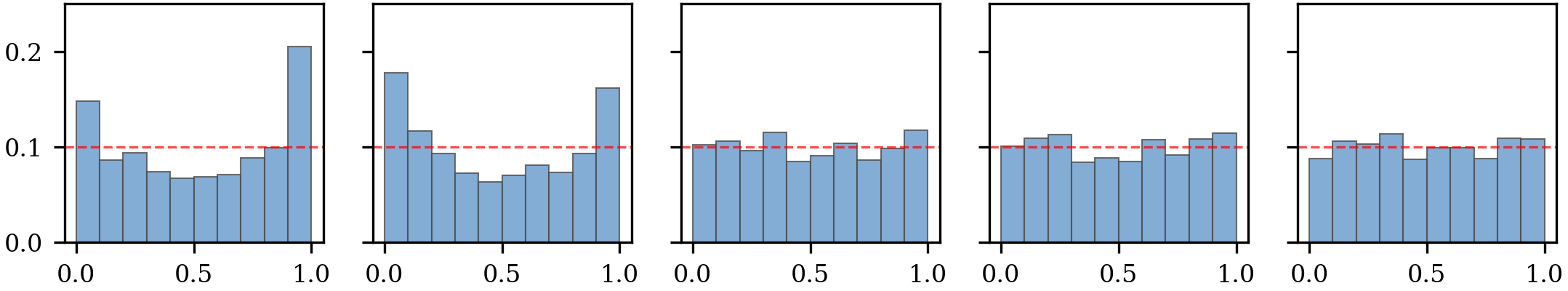}
\end{tabular}
\vspace{-0.2cm}
\caption{Taiwan randomized PIT histograms for the seed-123 log-millimeter
analysis. The panels correspond to the SW, NW, SE, and NE cells. Within each
panel, the columns are Raw, BC(Global), BC(GAM), CPIT(Global), and CPIT(GAM).}
\label{fig:wb-tw-pit}
\end{figure}


\section{Discussion}
\label{sec:discussion}

CPIT is a distributional calibration layer for sample-only predictors, not a fixed-level conformal interval method. Its primary output is a calibrated predictive CDF, from which one can compute threshold probabilities, quantiles at arbitrary levels, HDR intervals, tail probabilities, tail expectations, calibrated resamples, and decision-relevant risk summaries. This distinction is important in applications such as precipitation forecasting, where the same predictive law may be queried at many thresholds and risk levels, not just at one nominal coverage level. The finite-sample guarantee for CPIT is therefore stated in terms of calibrated PIT ranks, reflecting its goal of distributional reliability. When exact marginal coverage is also required at a specified level, the PIT-centrality conformal wrapper in Theorem~\ref{thm:pit_central_interval} can be applied with a separate interval-calibration split to obtain a standard split-conformal prediction interval.

CPIT is also computationally simple. The calibrated predictive law is a weighted empirical distribution supported on the original generator draws, with optional smoothing for interpolation and more stable tail summaries. This makes the method easy to apply to black-box ensembles or simulation-based predictors, since it requires only predictive samples and a held-out calibration set. It also makes the graphical diagnostics transparent through PIT histograms, calibrated rank CDFs, interval-length comparisons, and threshold-event score curves, all of which evaluate different projections of the same calibrated distribution. This combination of reusable distributional output and graphical diagnostics is a key practical advantage of the method.

The number of generator draws $m$ determines the resolution of the empirical CDF and the amount of available tail support. Larger $m$ generally improves empirical quantiles, HDR intervals, calibrated resampling, and tail-risk summaries by providing a richer sample cloud. When $m$ is small, or when the main targets are high-coverage intervals or rare-event probabilities, the smoothed weighted CDF can be more stable than the raw weighted empirical CDF. Smoothing should therefore be viewed as a numerical regularization step, not as a replacement for adequate ensemble diversity.

At the same time, CPIT cannot create information that is absent from the underlying sample cloud. If the generator misses important modes, has too few tail samples, or produces samples on a coarse finite support, the calibrated CDF can reweight and smooth those samples but cannot fully reconstruct the missing conditional distribution. This limitation is most visible for empirical equal-tail and HDR intervals when the number of generator draws $m$ is small, or when very high coverage levels require extrapolation beyond the available sample support. For this reason, interval coverage, interval length, PIT calibration, CRPS, and threshold-event scores should be reported together. Good performance on one summary need not imply good distributional calibration.

The present theory gives marginal calibration under exchangeability. This is
the natural distribution-free target for a general post-processing layer, but
it does not imply exact conditional calibration at each covariate value.
More localized versions of CPIT could use covariate-dependent calibration maps, Mondrian partitions \citep{pmlr-v152-bostrom21a}, or weighted calibration samples.
Naive kernel- or nearest-neighbor-weighted PIT recalibration need not retain exact finite-sample marginal validity.
Carefully constructed localized conformal procedures, however, can preserve finite-sample marginal coverage while improving local adaptivity \citep{guan2023localized}.
Exact distribution-free conditional coverage at every covariate value remains impossible without additional restrictions \citep{barber2021limits}.

Another important direction is calibration under distribution shift. In many simulation-to-real and forecasting problems, the conditional distribution of the response given the predictive sample may be relatively stable, while the marginal distribution of covariates changes between calibration and deployment. Weighted conformal methods under covariate shift provide a natural route for adapting CPIT in this setting \citep{tibshirani2019covshift}. The same idea could be applied at the PIT level by weighting calibration cases based on their relevance to the deployment covariate distribution.

Overall, CPIT provides a computationally lightweight way to convert sample-only predictive output into a calibrated predictive distribution. Its main advantage is coherence. All reported quantities are derived from one CDF, so quantiles, intervals, exceedance probabilities, and resamples are mutually consistent. This makes CPIT especially useful as a post-processing step for modern ensemble, Monte Carlo, and generative prediction systems, where the predictor naturally returns samples, but downstream statistical analysis requires calibrated distributional summaries.



\section*{Data Availability Statement}

The WeatherBench 2 data used in this study, including the ECMWF IFS-ENS
forecasts and the corresponding ERA5 verification fields, are publicly
available through the WeatherBench 2 data archive:
\url{https://weatherbench2.readthedocs.io/en/latest/data-guide.html}.

\section*{Disclosure Statement}

No potential conflict of interest was reported by the authors.






\appendix
\section{Proofs for Section~\ref{sec:theory}}
\label{app:proofs}

\begin{proof}[Proof of Theorem~\ref{thm:bccpit}]
Let
\[
\mathcal J=\mathcal I_{\mathrm cal}\cup\{n+1\},
\qquad |\mathcal J|=N+1.
\]
Conditional on $\mathcal D_{\mathrm tr}$ and $\mathcal D_{\mathrm bias}$, the fitted
generator and bias-correction rule are fixed. By assumption,
$\{(\mathcal O_i,V_i):i\in\mathcal J\}$ is exchangeable. The randomized PIT
$u_i$ is the same measurable function of $(\mathcal O_i,V_i)$ for every
$i\in\mathcal J$, so $\{u_i:i\in\mathcal J\}$ is exchangeable. Moreover, each
$u_i$ lies in $(0,1)$ almost surely.

For $r\in\mathcal J$, define its upper rank among the $N+1$ PIT values by
\[
R_r^+
=
\sum_{k\in\mathcal J}\mathbf 1\{u_k\le u_r\}.
\]
Because $u_{n+1}>0$ almost surely, \eqref{eq:calibrator} gives
\[
\widehat C(u_{n+1})
=
\frac{1+\sum_{i\in\mathcal I_{\mathrm cal}}
\mathbf 1\{u_i\le u_{n+1}\}}{N+1}
=
\frac{R_{n+1}^+}{N+1}.
\]
For any $k\in\{1,\ldots,N+1\}$, at most $k$ indices can have upper rank no
larger than $k$. Exchangeability therefore implies
\begin{align*}
\mathbb P\{R_{n+1}^+\le k
\mid\mathcal D_{\mathrm tr},\mathcal D_{\mathrm bias}\}
&=
\frac{1}{N+1}\sum_{r\in\mathcal J}
\mathbb P\{R_r^+\le k
\mid\mathcal D_{\mathrm tr},\mathcal D_{\mathrm bias}\}\\
&=
\frac{1}{N+1}
\mathbb E\!\left[
\sum_{r\in\mathcal J}\mathbf 1\{R_r^+\le k\}
\,\middle|\,
\mathcal D_{\mathrm tr},\mathcal D_{\mathrm bias}
\right]\le \frac{k}{N+1}.
\end{align*}
For $t\in[0,1]$, take $k=\lfloor(N+1)t\rfloor$. Then
\[
\mathbb P\!\left\{
\widehat C(u_{n+1})\le t
\,\middle|\,
\mathcal D_{\mathrm tr},\mathcal D_{\mathrm bias}
\right\}
\le
\frac{k}{N+1}
\le t,
\]
which proves \eqref{eq:finite_pit_calibration}.

Conditional on all forecast-response objects, each $u_i$ is a strictly
increasing affine function of the independent continuous variable $V_i$.
Hence the $N+1$ PIT values are almost surely distinct. Their ranks are therefore
a uniformly random permutation of $\{1,\ldots,N+1\}$, so $R_{n+1}^+$ is uniform
on this set. It follows that $\widehat C(u_{n+1})$ is uniform on
$\{1/(N+1),\ldots,1\}$.
\end{proof}

\begin{proof}[Proof of Theorem~\ref{thm:pit_central_interval}]
Conditional on $\mathcal D_{\mathrm CPIT}$, the score rule
$(x,y,\widehat y^{(1)},\ldots,\widehat y^{(m)})\mapsto
\big|2\widetilde F_x(y)-1\big|$ is fixed and is applied identically to every case.
Thus
\[
\{S_i^{\mathrm pit}:i\in\mathcal I_{\mathrm int}\cup\{n+1\}\}
\]
is exchangeable. By \eqref{eq:pit-centrality-set},
\[
Y_{n+1}\in\mathcal C_\alpha^{\mathrm pit}(X_{n+1})
\quad\Longleftrightarrow\quad
S_{n+1}^{\mathrm pit}\le q_\alpha^{\mathrm pit}.
\]
The standard split-conformal rank argument gives
\[
\mathbb P\{S_{n+1}^{\mathrm pit}\le q_\alpha^{\mathrm pit}
\mid\mathcal D_{\mathrm CPIT}\}
\ge
\frac{k_\alpha}{N_{\mathrm int}+1}
\ge 1-\alpha,
\]
with the convention that the event is certain when $k_\alpha>N_{\mathrm int}$ and
$q_\alpha^{\mathrm pit}=+\infty$. This proves \eqref{eq:pit_central_coverage}.

If the $N_{\mathrm int}+1$ scores are almost surely distinct, the rank of the test
score is uniform on $\{1,\ldots,N_{\mathrm int}+1\}$. The coverage probability is
then exactly $k_\alpha/(N_{\mathrm int}+1)$, which is at most
$1-\alpha+1/(N_{\mathrm int}+1)$.

Finally, $\alpha_1<\alpha_2$ implies
$k_{\alpha_1}\ge k_{\alpha_2}$ and therefore
$q_{\alpha_1}^{\mathrm pit}\ge q_{\alpha_2}^{\mathrm pit}$. The corresponding score
sublevel sets are nested:
$\mathcal C_{\alpha_1}^{\mathrm pit}(x)\supseteq
\mathcal C_{\alpha_2}^{\mathrm pit}(x)$ for every $x$.
\end{proof}

\begin{proof}[Proof of Proposition~\ref{prop:smoothed_regular}]
Fix $x$ and write
\[
z_j=\widehat Y_{\mathrm{adj},(j)}(x),
\qquad j=1,\ldots,m.
\]
The weights $w_1,\ldots,w_m$ are nonnegative and sum to one because they are
increments of the nondecreasing finite-support calibration map $\widehat C_m$.
Consequently,
\[
\widetilde F_x^\tau(y)
=
\sum_{j=1}^m w_j
\Phi\!\left(\frac{y-z_j}{\tau_x}\right)
\]
is a convex combination of Gaussian CDFs and hence is a proper CDF.
Differentiating term by term gives \eqref{eq:smoothed_density_theory}. At least
one weight is positive, and every Gaussian density is strictly positive on
$\mathbb R$, so $\widetilde f_x^\tau(y)>0$ for every $y$. Thus the distribution
has full support and its CDF is strictly increasing.

Let $J$ be a discrete random index with $\mathbb P(J=j)=w_j$, and let
$Z\sim N(0,1)$ be independent of $J$. Define
\[
\widetilde Y_x=z_J,
\qquad
\widetilde Y_x^\tau=z_J+\tau_x Z.
\]
Then $\widetilde Y_x\sim\widetilde P_x$ and
$\widetilde Y_x^\tau\sim\widetilde P_x^\tau$, so this construction is a coupling
of the two distributions. Therefore
\[
W_1(\widetilde P_x^\tau,\widetilde P_x)
\le
\mathbb E|\widetilde Y_x^\tau-\widetilde Y_x|
=
\tau_x\mathbb E|Z|
=
\tau_x\sqrt{2/\pi}.
\]
If $h$ is $L$-Lipschitz, the same coupling gives
\[
\left|
\mathbb E_{\widetilde P_x^\tau}\{h(Y)\}
-
\mathbb E_{\widetilde P_x}\{h(Y)\}
\right|
\le
\mathbb E\left|h(\widetilde Y_x^\tau)-h(\widetilde Y_x)\right|
\le
L\tau_x\sqrt{2/\pi}.
\]
\end{proof}

\begin{proof}[Proof of Theorem~\ref{thm:dkw_stability}]
Let
\[
\mathbb C_N(u)
=
\frac{1}{N}\sum_{i\in\mathcal I_{\mathrm cal}}
\mathbf 1\{u_i\le u\}
\]
be the ordinary empirical CDF of the calibration PIT values. Conditional on
$\mathcal D_{\mathrm tr}$ and $\mathcal D_{\mathrm bias}$, these values are independent
with common CDF $C^\star$. The Dvoretzky-Kiefer-Wolfowitz (DKW) inequality
\citep{dvoretzky1956asymptotic,massart1990tight} gives
\[
\mathbb P\!\left\{
\sup_{u\in[0,1]}|\mathbb C_N(u)-C^\star(u)|>\varepsilon
\,\middle|\,
\mathcal D_{\mathrm tr},\mathcal D_{\mathrm bias}
\right\}
\le 2\exp(-2N\varepsilon^2).
\]
For $u>0$,
\[
\widehat C(u)=\frac{1+N\mathbb C_N(u)}{N+1},
\]
and therefore
\[
|\widehat C(u)-C^\star(u)|
\le
|\mathbb C_N(u)-C^\star(u)|+\frac{1}{N+1}.
\]
At $u=0$, both CDFs are zero. Hence, on the event
\[
\mathcal E_\varepsilon
=
\left\{
\sup_{u\in[0,1]}|\mathbb C_N(u)-C^\star(u)|\le\varepsilon
\right\},
\]
we have
\begin{equation}
\Delta_N
:=
\sup_{u\in[0,1]}|\widehat C(u)-C^\star(u)|
\le
\delta_{N,\varepsilon}.
\label{eq:proof_calibration_sup}
\end{equation}

On $\mathcal E_\varepsilon$, the assumption
$\delta_{N,\varepsilon}<C^\star\{m/(m+1)\}$ implies
\[
\widehat C\{m/(m+1)\}
\ge
C^\star\{m/(m+1)\}-\Delta_N
>0.
\]
For $j=0,\ldots,m$, define the cumulative weight error
\begin{align*}
B_j
&=
\sum_{\ell=1}^j(w_\ell-w_\ell^\star)=
\frac{\widehat C\{j/(m+1)\}}
     {\widehat C\{m/(m+1)\}}
-
\frac{C^\star\{j/(m+1)\}}
     {C^\star\{m/(m+1)\}},
\end{align*}
where the empty sum is zero. Thus $B_0=B_m=0$. For every $j=0,\ldots,m$,
\begin{align*}
|B_j|
&\le
\frac{|\widehat C\{j/(m+1)\}-C^\star\{j/(m+1)\}|}
     {\widehat C\{m/(m+1)\}}\\
&\quad+
C^\star\{j/(m+1)\}
\left|
\frac{1}{\widehat C\{m/(m+1)\}}
-
\frac{1}{C^\star\{m/(m+1)\}}
\right|\\
&\le
\frac{2\Delta_N}{\widehat C\{m/(m+1)\}}\\
&\le
\frac{2\Delta_N}
{C^\star\{m/(m+1)\}-\Delta_N}.
\end{align*}
The second inequality uses the monotonicity of $C^\star$, which gives
$C^\star\{j/(m+1)\}\le C^\star\{m/(m+1)\}$.

Fix $x$ and $y$, and write
\[
K_j(x,y)
=
\Phi\!\left(
\frac{y-\widehat Y_{\mathrm{adj},(j)}(x)}{\tau_x}
\right).
\]
Because the adjusted samples are ordered,
$K_1(x,y)\ge\cdots\ge K_m(x,y)$. Also,
$w_j-w_j^\star=B_j-B_{j-1}$. Summation by parts therefore yields
\begin{align*}
\widetilde F_x^\tau(y)-F_{x,C^\star}^\tau(y)
&=
\sum_{j=1}^m(B_j-B_{j-1})K_j(x,y)=
\sum_{j=1}^{m-1}B_j\{K_j(x,y)-K_{j+1}(x,y)\}.
\end{align*}
Since the differences $K_j-K_{j+1}$ are nonnegative and telescope to at most
one,
\[
\left|\widetilde F_x^\tau(y)-F_{x,C^\star}^\tau(y)\right|
\le
\max_{0\le j\le m}|B_j|
\le
\frac{2\Delta_N}
{C^\star\{m/(m+1)\}-\Delta_N}.
\]
On $\mathcal E_\varepsilon$, \eqref{eq:proof_calibration_sup} and the monotonicity
of $z\mapsto 2z/[C^\star\{m/(m+1)\}-z]$ give the bound in
\eqref{eq:dkw_stability}, uniformly over $x\in\mathcal X_0$ and
$y\in\mathbb R$. The DKW inequality bounds the probability of
$\mathcal E_\varepsilon^c$ by $2\exp(-2N\varepsilon^2)$.

For the unsmoothed CDFs, replace $K_j(x,y)$ by
$\mathbf 1\{\widehat Y_{\mathrm{adj},(j)}(x)\le y\}$. This sequence is again
nonincreasing in $j$, so the same summation-by-parts argument applies.
\end{proof}

\end{document}